\documentclass[envcountsect,runningheads]{llncs}
\usepackage[utf8]{inputenc}

\usepackage{algorithmic}
\usepackage{algorithm}
\usepackage{amsmath}
\usepackage{amsxtra}
\usepackage{epsfig}
\usepackage{float}
\usepackage{authblk}
\usepackage{url}
\usepackage{mathptmx} 
\usepackage{enumerate}
\usepackage[numbers,sectionbib]{natbib}
\usepackage{fancyvrb}
\usepackage{marginnote}
\usepackage{mathtools}
\usepackage{hyperref}

\usepackage{todonotes}

\usepackage{blkarray}
\usepackage{subcaption}
\usepackage{tikz}
\usetikzlibrary{automata,positioning}

\usepackage{amssymb}

\newcommand{\ignore}[1]{}

\newcommand{\asymeq}{=_{\tiny{\downarrow}}^?}

\newtheorem{appxlem}{Lemma}[section]

\newtheorem{corol}{Corollary}[appxlem]

\newcommand{\xddots}{%
  \raise 4pt \hbox {.}
  \mkern 6mu
  \raise 1pt \hbox {.}
  \mkern 6mu
  \raise -2pt \hbox {.}
}

\DefineVerbatimEnvironment{todos}
 {Verbatim}
 {fontsize=\footnotesize}

\title{On Asymmetric Unification for the Theory of XOR with a Homomorphism} 

\author{Christopher Lynch\inst{1} \and
	Andrew M. Marshall\inst{2} \and
	Catherine Meadows\inst{3} \and
	Paliath Narendran\inst{4} \and
	Veena Ravishankar\inst{2}
}
\authorrunning{ Lynch, Marshall, Meadows, Narendran, Ravishankar}

\institute{Clarkson University,  Potsdam, NY, U.S.A.\\
	\email{clynch@clarkson.edu} 
	\and
	University of Mary Washington, Fredericksburg, VA, U.S.A.\\ 
	\email{marshall@umw.edu, vravisha@umw.edu} 
	\and
	Naval Research Laboratory, Washington, D.C., U.S.A.\\ \email{catherine.meadows@nrl.navy.mil} 
	\and
	University at Albany-SUNY, Albany, NY, U.S.A.\\ \email{pnarendran@albany.edu}
}

\date{\today}

\usepackage{tikz}
\usetikzlibrary{automata,positioning}

\begin{document}	

\maketitle

\begin{abstract}
Asymmetric unification, or unification with irreducibility
constraints, is a newly developed paradigm that arose out of the
automated analysis of cryptographic protocols.  However, there are
still relatively few asymmetric unification algorithms. In this paper
we address this lack by exploring the application of automata-based
unification methods. We examine the theory of xor with a homomorphism,
ACUNh, from the point of view of asymmetric unification, and develop a
new automata-based decision procedure. Then, we adapt a recently developed
asymmetric combination procedure to produce a general asymmetric-ACUNh
decision procedure. Finally, we present a new approach for
obtaining a solution-generating asymmetric-ACUNh unification
automaton.  We also compare our approach to the most commonly used
form of asymmetric unification available today, variant unification.
\end{abstract}

\section{Introduction}
We examine the newly developed paradigm of asymmetric unification in
the theory of \emph{xor} with a homomorphism.  Asymmetric unification
is motivated by  requirements arising from
symbolic cryptographic protocol analysis~\cite{Erbatur13a}.
These symbolic analysis
methods require unification-based exploration of a space in which the
states obey rich equational theories that can be expressed as a
decomposition $R \uplus \Delta$ , where $R$ is a set of rewrite rules that is
confluent, terminating, and coherent modulo $\Delta$. However, in order to
apply state space reduction techniques, it is usually necessary for at
least part of this state to be in normal form, and to remain in normal
form even after unification is performed.  This requirement can be
expressed as an {\em asymmetric} unification problem $\{ s_1
=_{\downarrow} t_1, ~\ldots, ~s_n =_{\downarrow} t_n \} $ where the
$=_{\downarrow}$ denotes a unification problem with the restriction
that any unifier leaves the right-hand side of each equation
irreducible.

At this point there are relatively few such  algorithms.  Thus in most cases when asymmetric unification is needed,
an algorithm based on \emph{variant unification}~\cite{DBLP:journals/entcs/EscobarMS09} is used.  Variant unification turns an $R \uplus \Delta$-problem into a set of $\Delta$-problems.
Application of  variant unification requires that a number of conditions on the decomposition be satisfied.  In particular,  the set of $\Delta$-problems produced must always be finite (this is equivalent to the \emph{finite variant property}~\cite{comon-delaune}) and $\Delta$-unification must be decidable and finitary.   Unfortunately, there is a class of theories commonly occurring in cryptographic protocols that do not have decompositions satisfying these necessary conditions: theories including an operator $h$ that is homomorphic over an Abelian group operator $+$, that is $AGh$.  There are a number of cryptosystems that include an operation that is homomorphic over an Abelian  group operator, and a number of constructions that rely on this homomorphic property.  These include for example RSA \cite{DBLP:journals/cacm/RivestSA78}, whose homomorphic property is used in Chaum's blind signatures \cite{DBLP:conf/crypto/Chaum83}, and Pallier cryptosystems \cite{DBLP:conf/eurocrypt/Paillier99}, used in electronic voting and digital cash protocols. Thus an alternative approach is called for. 

In this paper we concentrate on asymmetric unification for a special case of $AGh$: 
the theory of \emph{xor} with homomorphism, or $ACUNh$. 
We first develop an automata-based $ACUNh$-asymmetric decision procedure.
We then apply a recently developed combination procedure for asymmetric unification algorithms
to obtain a general asymmetric decision procedure allowing for free function symbols.
This requires a non-trivial adaptation of the combination procedure, which originally required that the algorithms combined were not only decision procedures but produced complete sets of unifiers.  In addition, the decomposition of $ACUNh$ we use is $\Delta = ACh$.  It is known that unification modulo $ACh$  is undecidable~\cite{DBLP:conf/lics/Narendran96}, so our result also yields the first asymmetric decision procedure for which $\Delta$ does not have a decidable finitary unification algorithm. 

\ignore{We then consider the problem of producing complete sets of asymmetric unifiers for $ACUNh$.
We first show, via an example, that asymmetric unification modulo $ACUNh$ is not finitary. 
We then discuss how the decision procedure developed in this paper
could be adapted to produce an automaton that generates a (possibly infinite) complete set of solutions. We outline this
promising new approach to conclude the paper and point to future work.}
We then consider the problem of producing complete sets of asymmetric unifiers for $ACUNh$. We show how the decision procedure developed in this paper can be adapted to produce an automaton that generates a (possibly infinite) complete set of solutions. We then show, via an example, that asymmetric unification modulo $ACUNh$ is not finitary.    

\subsection{Outline}
Section~\ref{Sec:Prelims} provides a brief description of preliminaries. 
Section~\ref{Sec:Decpro} develops an automaton based decision
procedure for the $ACUNh$-theory. In Section~\ref{Sec:Subpro} an
automaton approach that produces substitutions is outlined. Section~\ref{Sec:Combination} 
develops the modified combination method needed to obtain general asymmetric algorithms. In Section~\ref{sec:conclusion} we conclude the paper and discuss further work.

\section{Preliminaries}\label{Sec:Prelims}

We use the standard notation of equational
unification~\cite{BaaderSnyd-01} and term rewriting
systems~\cite{BaaderSnyd-01}.
$\Sigma$-terms, denoted by $T(\Sigma,$ $ \mathcal{X})$, are built over the
signature $\Sigma$ and the (countably infinite) set of variables
$\mathcal{X}$.  The terms $t|_p$ and $t[u]_p$ denote respectively the
subterm of $t$ at the position $p$, and the term $t$ having $u$ as
subterm at position $p$. The symbol of $t$ occurring at the position
$p$ (resp. the top symbol of $t$) is written $t(p)$ (resp.
$t(\epsilon)$).  The set of positions of a term $t$ is denoted by
$Pos(t)$, the set of non variable positions for a term $t$ over a
signature $\Sigma$ is denoted by $Pos(t)_{\Sigma}$.  A
\emph{$\Sigma$-rooted} term is a term whose top symbol is in $\Sigma$.
The set of variables of a term $t$ is denoted by $Var(t)$.  A term is
\emph{ground} if it contains no variables.  

\begin{definition}
	Let $\Gamma$ be an $E$-unification problem, let $\mathcal{X}$ denote the set of variables occurring 
	in $\Gamma$ and $\mathcal{C}$ the set of free constants occurring in $\Gamma$. For a given linear ordering $<$ on
	$\mathcal{X} \cup \mathcal{C}$, and for each $c \in \mathcal{C}$ define the set $V_c$ as $\{ x ~| ~x \text{ is a variable with } x < c \}$.
	An \emph{$E$-unification problem with linear constant restriction} (LCR) is an 
	$E$-unification problem with constants, $\Gamma$, where each constant $c$ in $\Gamma$
	is equipped with a set $V_c$ of variables. A solution of the problem is an $E$-unifier $\sigma$
	of $\Gamma$ such that for all $c, x$ with $x \in V_c$, the constant $c$ does not occur
	in $x \sigma$. We call $\sigma$ an \emph{$E$-unifier with LCR.}
\end{definition}

A \emph{rewrite rule} is an ordered pair $l \rightarrow r$ such that
$l, r \in T(\Sigma, \mathcal{X})$ and $l \not\in \mathcal{X}$.  We use
$R$ to denote a term rewrite system which is defined as a set of
rewrite rules.  The rewrite relation on $T(\Sigma, \mathcal{X})$,
written $t \rightarrow_R s$, hold between $t$ and $s$ iff there exists
a non-variable $p \in Pos_{\Sigma}(t)$, $l \rightarrow r \in R$ and a
substitution $\sigma$, such that $t |_p = l \sigma$ and $s =t[r
  \sigma]_p$.  The relation $\rightarrow_{R/E}$ on $T (\Sigma,
\mathcal{X})$ is $=_{E} \circ \rightarrow_{R} \circ =_{E}$.  The
relation $\rightarrow_{R,E}$ on $T (\Sigma, \mathcal{X})$ is defined
as: $t \rightarrow_{R,E} t'$ if there exists a position $p \in
Pos_{\Sigma}(t)$, a rule $l \rightarrow r \in R$ and a substitution
$\sigma$ such that $t|_{p} =_{E} l \sigma$ and $t' = t[r \sigma]_{p}$.
The transitive (resp. transitive and reflexive) closure of
$\rightarrow_{R,E}$ is denoted by $\rightarrow^{+}_{R,E}$
(resp. $\rightarrow^{*}_{R,E}$).  A term $t$ is \emph{$\rightarrow_{R,
    E}$ irreducible} (or in \emph{$R,E$-normal form}) if there is no
term $t'$ such that $t \rightarrow_{R,E} t'$.  If $\rightarrow_{R,E}$
is confluent and terminating we denote the irreducible version of a
term, $t$, by $t \downarrow_{R,E}$.

\begin{definition}
	We call $(\Sigma, ~E, ~R)$ a \emph{weak decomposition} of an equational theory $\Delta$ over
	a signature $\Sigma$ if $\Delta = R \uplus E$ and $R$ and $E$ satisfy the following conditions:
	\begin{enumerate}
		\item Matching modulo~$E$ is decidable.
		\item $R$ is terminating modulo $E$, i.e., the relation
		  $\rightarrow_{R/E}$ is terminating.
		\item The relation $\rightarrow_{R,E}$ is confluent and $E$-coherent, i.e., $\forall t_1, t_2, t_3$ if $t_1 \rightarrow_{R,E} t_2$ and 
		$t_1 =_{E} t_3$ then $\exists ~t_4, t_5$ such that $t_2 \rightarrow^{*}_{R, E} t_4$, $t_3 \rightarrow^{+}_{R, E} t_5$, and
		$t_4 =_{E} t_5$.  
	\end{enumerate}
\end{definition}

This definition is a modification of the definition in~\cite{Erbatur13a}.
where
asymmetric unification and the corresponding theory decomposition are
first defined.  The last restrictions ensure that $s
\rightarrow^{!}_{R/E} t$ iff $s \rightarrow^{!}_{R,E} t$
(see~\cite{DBLP:journals/entcs/EscobarMS09,Erbatur13a}).

\begin{definition}[Asymmetric Unification]\label{asym_prob_def}
	Given a weak decomposition $(\Sigma, E, R)$
	of an equational theory, a substitution $\sigma$
	is an \emph{asymmetric $R,E$-unifier} of a set $\mathcal{S}$ of asymmetric equations
	$\{ s_1 =_{\downarrow} t_1, ~\ldots, ~s_n =_{\downarrow} t_n \} $ iff for each
	asymmetric equations $s_i =_{\downarrow} t_i$, $\sigma$ is an $(E \cup R)$-unifier
	of the equation $s_i =^{?} t_i$ and $(t_i \downarrow_{R,E})\sigma$ is in $R,E$-normal
	form. A set of substitutions $\Omega$ is a \emph{complete set of asymmetric} 
	$R, E$-unifiers of $\mathcal{S}$ 
	(denoted $CSAU_{R \cup E}(\mathcal{S})$ or just $CSAU(\mathcal{S})$ if the background theory is clear) iff: 
	(i) every member of $\Omega$ is an asymmetric 
	$R,E$-unifier of $\mathcal{S}$, and (ii) for every asymmetric $R,E$-unifier $\theta$ of
	$\mathcal{S}$ there exists a $\sigma \in \Omega$ such that $\sigma \leq_{E}^{Var(\mathcal{S})} \theta$.
\end{definition}

\begin{example}
	Let $R = \{ x \oplus 0 \rightarrow x, ~x \oplus x \rightarrow 0 , ~x \oplus x \oplus y \rightarrow y \}$
	and $E$ be the $AC$ theory for $\oplus$. Consider the equation $y \oplus x =_{\downarrow} x \oplus a$,
	the substitution $\sigma_1 = \{ y \mapsto a \}$ is an asymmetric solution but, $\sigma_2 = \{ x \mapsto 0, ~y \mapsto a \}$ is not.
\end{example}

\begin{definition}[Asymmetric Unification with Linear Constant Restriction]
	Let $\mathcal{S}$ be a set of asymmetric equations with some LCR.
	A substitution $\sigma$ is an \emph{asymmetric $R,E$-unifier} of $\mathcal{S}$ with LCR  iff $\sigma$ is an asymmetric solution to $\mathcal{S}$ and $\sigma$ satisfies the LCR.
\end{definition}

\begin{definition}\label{def1}
  Let $R$ be a term rewriting system and ${E}$~be a set of identities.
  We say $(R, E)$ is 
  \emph{$R, E$-convergent} if and only if
  \begin{itemize}
  \item[(a)] $\rightarrow_{R, E}^{} \,$ is terminating, \emph{and}

  \item[(b)] for all terms $s$, $t$, if $s \, \approx_{R \cup E}^{} \, t$, there exist
  terms $s^{\prime}$, $t^{\prime}$ such that $
  s \, \rightarrow_{R, E}^{!} \, s^{\prime} , \;
  t \, \rightarrow_{R, E}^{!} \, t^{\prime} , \; \;
  \mathrm{and} \; \;
  s^{\prime} \, \approx_E^{} \, t^{\prime} $
  \end{itemize}
\end{definition}
\begin{definition}
  A term $t$ is an $R, \Delta$-normal form of a term $s$ if and only if $\mbox{$s \rightarrow_{R, \Delta}^{!}t$}$.
  This is often represented as $t \, = \, s\!\big\downarrow_{R, \Delta}$. 
\end{definition}

\section{An Asymmetric $ACUNh$-unification Decision Procedure via an Automata 
	Approach}~\label{Sec:Decpro}
\noindent
In this section we develop a new asymmetric unification algorithm
for the theory $ACUNh$. The theory $ACUNh$ consists of the following identities:
		$x + x  \approx  0,\;
		x + 0  \approx  x,\;
		h(x+y)  \approx  h(x) + h(y),\;
		h(0)  \approx  0,\;
		(x + y) + z  \approx  x + (y + z),\;
		x + y  \approx  y + x$

\noindent
Following the definition of asymmetric unification, we first 
decompose the theory into a set of rewrite rules, $R$, modulo a set
of equations, $\Delta$. Actually, there are two such 
decompositions possible.
\noindent 
The first decomposition keeps~$associativity$ and ~$commutativity$ as identities ${\Delta}^{}$
and the rest as rewrite rules. This decomposition has the following $AC$-convergent
term rewriting system $R_1^{}$:
 $ x + x  \rightarrow  0,\;$
  $x + 0  \rightarrow  x,\;$
  $x + (y + x)  \rightarrow y,\;$
  $h(x+y)  \rightarrow  h(x) + h(y),\;$
  $h(0)  \rightarrow  0$,
 as well as~$R_1^{\prime}$:
$  x + x  \rightarrow  0, \;
  x + 0  \rightarrow  x, \;
  x + (y + x)  \rightarrow  y, \;
  h(x) + h(y)  \rightarrow  h(x + y), \;
  h(0)  \rightarrow  0$
(when $+$ is given a higher precedence over~$h$).


The second decomposition has $associativity$, $commutativity$ and the distributive
homomorphism identity as~$\Delta$~\footnote{This is the background theory.}, i.e., $\Delta = ACh$. Our goal here is to prove that 
the following
term rewriting system~$R_2^{}$:
  $x + x  \rightarrow  0,\;$
$  x + 0  \rightarrow  x,\;$
$  x + (y + x)  \rightarrow  y,\;$
  $h(0)  \rightarrow  0$
 is $ACh$-convergent.
The proof for convergence of $\rightarrow_{R_2^{}, \, ACh}^{}$ is provided in Appendix~$A$.

\noindent
Decidability of asymmetric unification
for the theory $R_2^{}$, $ACh$ can be shown by automata
-theoretic methods analogous to
the method used for deciding the Weak
Second Order Theory of One successor
(WS1S)~\cite{Elgot,buchi1960weak}.
In WS1S we consider quantification over finite sets of natural
numbers, along with one successor function. All equations or formulas
are transformed into finite-state automata which accepts the strings that
correspond to a model of the
formula~\cite{klaedtke2002parikh,vardi2008automata}. This automata-based
approach is key to showing decidability of~WS1S, since the
satisfiability of WS1S formulas reduces to the automata
intersection-emptiness problem. We follow the same approach here. 

To be precise, what we show here is that \emph{ground} solvability
of asymmetric unification, for a given set of constants, is decidable.
We explain at the end of this section why this is equivalent to solvability in general, in Lemma~\ref{Lemmatwo} and Lemma~\ref{Lemma 13}. 

\paragraph{Problems with one constant}: For ease of exposition, let us
consider the case where there is only one constant~$a$. Thus every
ground term can be represented as a set of natural numbers. The
homomorphism~$\mathsf{h}$ is treated as a successor
function. Just as in WS1S, the input to the automata are column vectors of bits. The length of each column vector
is the number of variables in the problem.
$\Sigma=\left\{ \vphantom{b^b}
          \begin{psmallmatrix}
           0 \\
           0 \\
          \vdots \\
           0 
          \end{psmallmatrix}, \ldots ,\begin{psmallmatrix}
           1 \\
           1 \\
           \vdots \\
           1 
          \end{psmallmatrix} \right\}$.
The deterministic finite automata (DFA) are illustrated here. 
The $\mathsf{+}$~operator behaves like the 
\emph{symmetric set difference} operator.

We illustrate how an automaton is constructed for each equation in standard
form. In order to avoid cluttering up the diagrams
the dead state has been included only for the first automaton.
The missing transitions lead to the dead state by default
for the others. Recall that we are considering the case of one constant~$a$. 

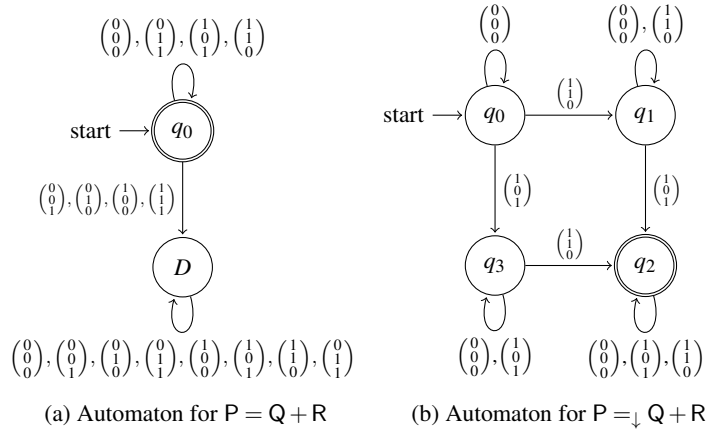
\begin{figure}[h]  
\centering 
  \begin{subfigure}[b]{0.4\linewidth}
   \begin{tikzpicture}[shorten >=1pt,node distance=1cm,auto]
   \node[state,initial,accepting] (q_0) {$q_0$};
   \node[state] (D) [below=of q_0] {$D$};
   \path[->]
    (q_0) edge [loop above] node  [scale=0.85]
    {$ \begin{psmallmatrix}
           0 \\
           0 \\
           0 
         \end{psmallmatrix}, \begin{psmallmatrix}
           0 \\
           1 \\
           1 
          \end{psmallmatrix}, \begin{psmallmatrix}
           1 \\
           0 \\
           1 
          \end{psmallmatrix}, \begin{psmallmatrix}
           1 \\
           1 \\
           0 
          \end{psmallmatrix}$} ()
             edge node [left] [scale=0.7]{$ \begin{psmallmatrix}
           0 \\
           0 \\
           1 
          \end{psmallmatrix}, \begin{psmallmatrix}
           0 \\
           1 \\
           0 
          \end{psmallmatrix}, \begin{psmallmatrix}
           1 \\
           0 \\
           0 
          \end{psmallmatrix}, \begin{psmallmatrix}
           1 \\
           1 \\
           1
         \end{psmallmatrix}$} (D)
       (D) edge [loop below]node [scale=0.85]{$\begin{psmallmatrix}
           0 \\
           0 \\
           0 
          \end{psmallmatrix},\begin{psmallmatrix}
           0 \\
           0 \\
           1 
          \end{psmallmatrix},\begin{psmallmatrix}
           0 \\
           1 \\
           0 
          \end{psmallmatrix},\begin{psmallmatrix}
           0 \\
           1 \\
           1 
          \end{psmallmatrix},\begin{psmallmatrix}
           1 \\
           0 \\
           0 
          \end{psmallmatrix},\begin{psmallmatrix}
           1 \\
           0 \\
           1 
          \end{psmallmatrix},\begin{psmallmatrix}
           1 \\
           1 \\
           0 
          \end{psmallmatrix},\begin{psmallmatrix}
           0 \\
           1 \\
           1 
          \end{psmallmatrix}$} ();  
\end{tikzpicture}
    \caption{Automaton for $\mathsf{P=Q+R}$} \label{fig:M1} 
  \end{subfigure}
\begin{subfigure}[b]{0.4\linewidth}
\begin{tikzpicture}[shorten >=1pt,node distance=2cm,on grid,auto]
   \node[state,initial] (q_0) {$q_0$};  
   \node[state] (q_1) [right=of q_0] {$q_1$};
   \node[state] (q_3) [below=of q_0] {$q_3$};
   \node[state,accepting] (q_2) [below=of q_1] {$q_2$};
   \path[->]
   (q_0) edge [loop above] node [scale=0.85]
    {$ \begin{psmallmatrix}
           0 \\
           0 \\
           0 
         \end{psmallmatrix}$} ()
         edge node[scale=0.75] {$ \begin{psmallmatrix}
           1 \\
           0 \\
           1 
          \end{psmallmatrix}$} (q_3)
          edge node [scale=0.7]{$ \begin{psmallmatrix}
           1 \\
           1 \\
           0 
          \end{psmallmatrix}$} (q_1)
          (q_1) edge [loop above] node [scale=0.85]
    {$ \begin{psmallmatrix}
           0 \\
           0 \\
           0 
         \end{psmallmatrix},\begin{psmallmatrix}
           1 \\
           1 \\
           0 
         \end{psmallmatrix}$} ()
         edge node [scale=0.7]{$ \begin{psmallmatrix}
           1 \\
           0 \\
           1 
          \end{psmallmatrix}$} (q_2)
          (q_2) edge [loop below] node [scale=0.85]
    {$ \begin{psmallmatrix}
           0 \\
           0 \\
           0 
         \end{psmallmatrix}$,$ \begin{psmallmatrix}
           1 \\
           0 \\
           1 
         \end{psmallmatrix}$,$ \begin{psmallmatrix}
           1 \\
           1 \\
           0 
         \end{psmallmatrix}$} ()
         (q_3) edge [loop below] node [scale=0.85]
    {$ \begin{psmallmatrix}
           0 \\
           0 \\
           0 
         \end{psmallmatrix}$,$ \begin{psmallmatrix}
           1 \\
           0 \\
           1 
         \end{psmallmatrix}$} ()
         edge node [scale=0.7]
    {$ \begin{psmallmatrix}
           1 \\
           1 \\
           0 
         \end{psmallmatrix}$} (q_2);
 \end{tikzpicture}
 \caption{Automaton for $\mathsf{P=_{\downarrow}Q+R}$} \label{fig:M2}
\end{subfigure}
\caption{Automata Construction}\label{Auto1}
\end{figure}

\noindent
$\bf{Fig.~\ref{fig:M1}}$: Let $\mathsf{P_i, Q_i} \text{ and } \mathsf{R_i}$ denote 
the ${i^{th}}$ bits of $\mathsf{P, Q} \text{ and } \mathsf{R} \;respectively$. 
$\mathsf{P_i}$ has a value~1, when either
$\mathsf{Q_i}$ or $\mathsf{R_i}$ has a value 1. We need 3-bit alphabet
symbols for this equation. The input for this automaton are column vectors of 3-bits each, i.e., $\Sigma=\{  \begin{psmallmatrix}
           0 \\
           0 \\
           0 
         \end{psmallmatrix},\cdots,\begin{psmallmatrix}
           1 \\
           1 \\
           1 
         \end{psmallmatrix}\}$. For example, if
$\mathsf{R_2}$ = 0, $\mathsf{Q_2}$ = 1, then $\mathsf{P_2}$ = 1. The corresponding
alphabet symbol is $\begin{psmallmatrix}
           P_2 \\
           Q_2 \\
           R_2
          \end{psmallmatrix}$ = $\begin{psmallmatrix} 1 \\ 0 \\ 1
          \end{psmallmatrix}$. 
          Hence, only strings with the alphabet symbols $\{$ $\begin{psmallmatrix}
           0 \\
           0 \\
           0 
         \end{psmallmatrix}, \begin{psmallmatrix}
           0 \\
           1 \\
           1 
          \end{psmallmatrix}, \begin{psmallmatrix}
           1 \\
           0 \\
           1 
          \end{psmallmatrix}, \begin{psmallmatrix}
           1 \\
           1 \\
           0 
          \end{psmallmatrix}$ $\}$ are accepted by this automaton. 
The rest of the input symbols $\{$ $\begin{psmallmatrix}
           0 \\
           0 \\
           1 
         \end{psmallmatrix}, \begin{psmallmatrix}
           1 \\
           1 \\
           1 
          \end{psmallmatrix}, \begin{psmallmatrix}
           0 \\
           1 \\
           0 
          \end{psmallmatrix}, \begin{psmallmatrix}
           1 \\
           0 \\
           0 
          \end{psmallmatrix}$ $\}$ go to the dead state~$D$, as they violate 
the XOR property. 
Note that the string $\begin{psmallmatrix}
           1 \\
           0 \\
           1 
          \end{psmallmatrix} \begin{psmallmatrix}
           1 \\
           1 \\
           0 
          \end{psmallmatrix}$ is accepted by this automaton. This corresponds
to $\mathsf{P=a+h(a)}$,  
$\mathsf{Q=h(a)}$ and~$\mathsf{R=a}$.

\noindent
$\bf{Fig.~\ref{fig:M2}}$: To preserve asymmetry on the right-hand side of this equation, $\mathsf{Q+R}$ should be 
irreducible. If either $\mathsf{Q}$ or $\mathsf{R}$ is empty, or if they 
have any term in common, then a reduction will occur. For example, 
if $\mathsf{Q}$ = $\mathsf{h(a)}$ and $\mathsf{R}$ = $\mathsf{h(a)+a}$, 
there is a reduction, whereas if $\mathsf{R}$ = $\mathsf{h(a)}$ 
and $\mathsf{Q}$ = $\mathsf{a}$, irreducibility is preserved, since 
there is no common term and neither one is empty. 
Since neither $\mathsf{Q}$ nor $\mathsf{R}$ can be empty,
any accepted string should have one occurrence of $\begin{psmallmatrix}
           1 \\
           0 \\
           1 
         \end{psmallmatrix}$ and one occurrence of~$\begin{psmallmatrix}
           1 \\
           1 \\
           0 
         \end{psmallmatrix}$.
\begin{figure}[h]  
\centering 
  \begin{subfigure}[b]{0.4\linewidth}
  \begin{tikzpicture}[>=stealth,shorten >=1pt,auto,node distance=3cm]
    \node[state,initial,accepting] (q_0)      [scale=0.85]          {$q_0$};
    \node[state] (q_1) [right of = q_0] [scale=0.85]{$q_1$};
    \path[->] (q_0) edge [bend left]  node [scale=0.7]{$\begin{psmallmatrix}
           1 \\
           0 
          \end{psmallmatrix}$} (q_1)
          edge [loop above] node [scale=0.85]{$\begin{psmallmatrix}
           0 \\
           0 
          \end{psmallmatrix}$} ()
              (q_1) edge [bend left] node [scale=0.7]{$\begin{psmallmatrix}
           0 \\
           1 
          \end{psmallmatrix}$} (q_0)
          edge [loop above] node [scale=0.85]{$\begin{psmallmatrix}
           1 \\
           1 
          \end{psmallmatrix}$} ();
\end{tikzpicture}
\caption{Automaton for $\mathsf{X=h(Y)}$} \label{fig:M3}
  \end{subfigure}
\begin{subfigure}[b]{0.4\linewidth}
  \begin{tikzpicture}[shorten >=1pt,node distance=1.7cm,on grid,auto]
   \node[state,initial] (q_0) {$q_0$};
   \node[state] (q_1) [right=of q_0] {$q_1$};
   \node[state,accepting] (q_2) [below=of q_1] {$q_2$};
   \path[->]
    (q_0) edge node [scale=0.7]{$\begin{psmallmatrix}
           1 \\
           0 
          \end{psmallmatrix}$} (q_1)
          edge [loop above] node [scale=0.85]
    {$ \begin{psmallmatrix}
           0 \\
           0 
         \end{psmallmatrix}$} ()
        (q_1) edge[bend left] node [scale=0.7]{$\begin{psmallmatrix}
           0 \\
           1 
          \end{psmallmatrix}$} (q_2)
           edge [loop right] node[scale=0.85]{$\begin{psmallmatrix}
           1 \\
           1
          \end{psmallmatrix}$}()
          (q_2) edge [loop right] node[scale=0.85]{$\begin{psmallmatrix}
           0 \\
           0
          \end{psmallmatrix}$}()
          edge [bend left] node [scale=0.7]{$\begin{psmallmatrix}
           1 \\
           0
          \end{psmallmatrix}$} (q_1);
          \end{tikzpicture}
          \caption{Automaton for $\mathsf{X=_{\downarrow}h(Y)}$} \label{fig:M4}
\end{subfigure}
\caption{Automata construction}\label{Auto2}
\end{figure}
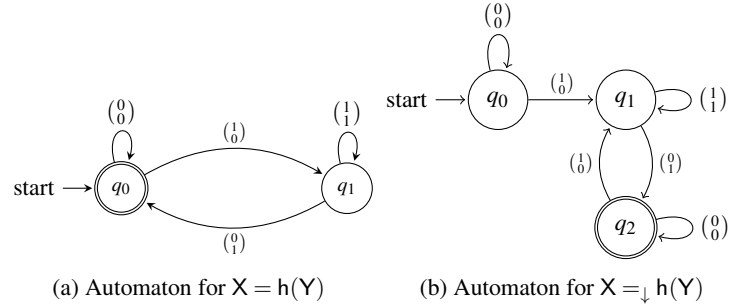

\noindent
$\bf{Fig.~\ref{fig:M3}}$: We need 2-bit vectors as alphabet symbols since we have 
two unknowns $\mathsf{X}$ and $\mathsf{Y}$. Remember that $\mathsf{h}$ acts like the successor function.
$q_0$ is the only accepting state. 
A state transition occurs with bit vectors $\begin{psmallmatrix}
           1 \\
           0         \end{psmallmatrix},\begin{psmallmatrix}
           0 \\
           1         \end{psmallmatrix}$. If $\mathsf{Y}$=1 in current state, then $\mathsf{X}$=1 in the next state, hence a transition occurs from $\mathsf{q_0}$ to $\mathsf{q_1}$, and vice versa. The ordering of variables is $\begin{psmallmatrix}
           Y \\
           X         \end{psmallmatrix}$.


\noindent
$\bf{Fig.~\ref{fig:M4}}$: In 
this equation, $\mathsf{h(Y)}$ should be in normal form. So
$\mathsf{Y}$ cannot be 0, but can contain terms of the form $\mathsf{u + v}$.
$\begin{psmallmatrix}
           Y \\
           X         \end{psmallmatrix}$ is the ordering of variables.
Therefore
the bit vector $\begin{psmallmatrix}
           1 \\
           0         \end{psmallmatrix}$ should be succeeded by $\begin{psmallmatrix}
           0 \\
           1         \end{psmallmatrix}$, with possible occurrences of the bit vector $\begin{psmallmatrix}
           1 \\
           1         \end{psmallmatrix}$ in between. Thus the string either ends with $\begin{psmallmatrix}
           0 \\
           1         \end{psmallmatrix}$ or $\begin{psmallmatrix}
           0 \\
           0         \end{psmallmatrix}$. 
For example, if
$\mathsf{Y}$ = $\mathsf{h(a)+a}$, then $\mathsf{X}$ = $\mathsf{h^2(a)+h(a)}$, which results in the string $\begin{psmallmatrix}
           1 \\
           0         \end{psmallmatrix}$$\begin{psmallmatrix}
           1 \\
           1         \end{psmallmatrix}$$\begin{psmallmatrix}
           0 \\
           1         \end{psmallmatrix}$ is accepted by this automaton.
           

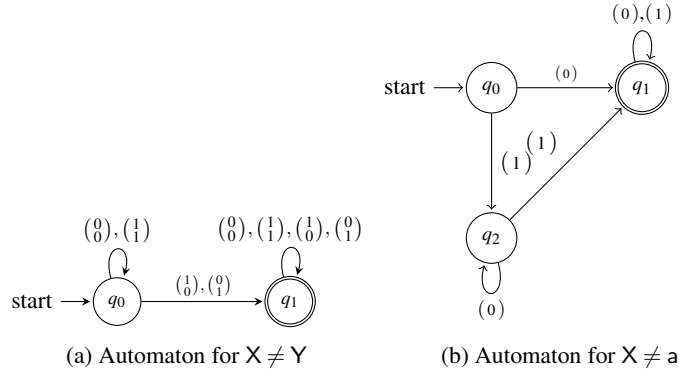
\begin{figure}[h]  
\centering 
  \begin{subfigure}[b]{0.4\columnwidth}
    \begin{tikzpicture}[>=stealth,shorten >=1pt,auto,node distance=2.3cm]
    \node[state,initial] (q_0)       [scale=0.8]         {$q_0$};
    \node[state,accepting] (q_1) [right of = q_0] [scale=0.8]{$q_1$};
    \path[->] (q_0) edge node [scale=0.7]{$\begin{psmallmatrix}
           1 \\
           0 
          \end{psmallmatrix},\begin{psmallmatrix}
           0 \\
           1
          \end{psmallmatrix}$} (q_1)
          edge [loop above] node[scale=0.85] {$\begin{psmallmatrix}
           0 \\
           0 
          \end{psmallmatrix},\begin{psmallmatrix}
           1 \\
           1
          \end{psmallmatrix}$} ()
              (q_1) edge [loop above] node[scale=0.85] {$\begin{psmallmatrix}
           0 \\
           0 
          \end{psmallmatrix},\begin{psmallmatrix}
           1 \\
           1
          \end{psmallmatrix},\begin{psmallmatrix}
           1 \\
           0 
          \end{psmallmatrix},\begin{psmallmatrix}
           0 \\
           1
          \end{psmallmatrix}$} ()
          ;
\end{tikzpicture}
\caption{Automaton for $\mathsf{X\not=Y}$} \label{fig:M4a}
  \end{subfigure}
\begin{subfigure}[b]{0.4\linewidth}
  \begin{tikzpicture}[shorten >=1pt,node distance=2cm,on grid,auto]
   \node[state,initial] (q_0) [scale=0.85]{$q_0$};
   \node[state,accepting] (q_1) [right=of q_0] [scale=0.85]{$q_1$};
   \node[state] (q_2) [below=of q_0] [scale=0.85]{$q_2$};
   \path[->]
    (q_0) edge node [scale=0.7]{$\begin{psmallmatrix}
           0 
          \end{psmallmatrix}$} (q_1)
          edge node 
    {$ \begin{psmallmatrix}
           1
         \end{psmallmatrix}$} (q_2)
        (q_1)
          edge[loop above] node [scale=0.85]{$\begin{psmallmatrix}
           0 
          \end{psmallmatrix}$,$\begin{psmallmatrix}
           1 
          \end{psmallmatrix}$} ()
          (q_2) edge [loop below] node [scale=0.85]
    {$ \begin{psmallmatrix}
           0 
         \end{psmallmatrix}$} ()
          edge node {$\begin{psmallmatrix}
           1 
          \end{psmallmatrix}$} (q_1);
          \end{tikzpicture}
          \caption{Automaton for $\mathsf{X\not=a}$} \label{fig:M4b}
\end{subfigure}
\caption{Automata construction}\label{Auto3}
\end{figure}
$\bf{Fig.~\ref{fig:M4a}}$: This automaton represents the disequality
$\mathsf{X^a\not=Y^a}$. In general, if there are two or more
constants, we have to guess which components are not equal. This
enables us to handle the disequality constraints mentioned in the next
section.


$\bf{Fig.~\ref{fig:M4b}}$: This automaton represents the disequality $\mathsf{X\not=a}$, where $a$ is a constant.
\begin{example}
Let $\Big\{ U=_{\downarrow}V+Y, \; W=h(V), \;Y=_{\downarrow}h(W) \Big\}$ be an asymmetric unification
problem.
We need 4-bit vectors and 3~automata since we have 4 unknowns in 3 equations, with bit-vectors represented in this ordering of set variables: $ \begin{psmallmatrix}
           V \\
            W\\
           Y \\
           U
         \end{psmallmatrix}$.
\noindent
We include the $\times$ (``don't-care'') symbol in state transitions to
indicate that the values can be either $0$ or~$1$. This is done to avoid
cluttering the diagrams. Note that here
this $\times$~symbol is
a placeholder for the variables which do not have any
significance in a given automaton. The automata constructed for this example are
indicated in~$\bf{Fig.~\ref{fig:M5}}$, $\bf{Fig.~\ref{fig:M6}}$ and $\bf{Fig.~\ref{fig:M7}}$.
The string~$\begin{psmallmatrix}
              1\\
           0 \\
            0 \\
           1
          \end{psmallmatrix}$$\begin{psmallmatrix}
              0\\
           1 \\
            0 \\
           0
          \end{psmallmatrix}$$\begin{psmallmatrix}
              0\\
           0 \\
            1 \\
           1
          \end{psmallmatrix}$$\begin{psmallmatrix}
              0\\
           0 \\
            0 \\
           0
          \end{psmallmatrix}$
is accepted by all the three automata.
The corresponding
asymmetric unifier is $ \{V \mapsto a, \, W\mapsto h(a), \, 
Y\mapsto h^2(a), \, U \mapsto ( h^2(a)+a ) \}. $    
\end{example}
\begin{figure}[h]  
\centering 

  \begin{subfigure}[b]{0.4\linewidth}
    \begin{tikzpicture}[shorten >=1pt,node distance=1.75cm,on grid,auto]
   \node[state,initial] (q_0) {$q_0$};
   \node[state] (q_1) [right=of q_0] {$q_1$};
   \node[state,accepting] (q_2) [below=of q_1] {$q_2$};
   \path[->]
    (q_0) edge node [scale=0.7]{$\begin{psmallmatrix}
           \times \\
           1 \\
           0\\
           \times
          \end{psmallmatrix}$} (q_1)
          edge [loop above] node [scale=0.85]
    {$ \begin{psmallmatrix}
           \times \\
           0 \\
           0\\
           \times
         \end{psmallmatrix}$} ()
        (q_1) edge [bend left] node[scale=0.7] {$\begin{psmallmatrix}
           \times \\
           0\\
           1\\
           \times 
          \end{psmallmatrix}$} (q_2)
          edge [loop above] node [scale=0.85]
    {$ \begin{psmallmatrix}
           \times \\
           1 \\
           1\\
           \times
         \end{psmallmatrix}$} ()
         (q_2) edge [loop below] node [scale=0.85]
    {$ \begin{psmallmatrix}
           \times \\
           0 \\
           0\\
           \times
         \end{psmallmatrix}$} ()
         edge [bend left] node [scale=0.7]{$\begin{psmallmatrix}
           \times \\
           1 \\
           0\\
           \times
          \end{psmallmatrix}$} (q_1);
          \end{tikzpicture}
          \caption{Automata for Example 3.1, Part 1}\label{fig:M5}
  \end{subfigure}
  \begin{subfigure}[b]{0.4\linewidth}
    \begin{tikzpicture}[shorten >=1pt,node distance=2cm,on grid,auto]
   \node[state,initial] (q_0) {$q_0$};  
   \node[state] (q_1) [right=of q_0] {$q_1$};
   \node[state] (q_3) [below=of q_0] {$q_3$};
   \node[state,accepting] (q_2) [below=of q_1] {$q_2$};
   \path[->]
   (q_0) edge [loop above] node [scale=0.85]
    {$ \begin{psmallmatrix}
           0 \\
           \times \\
           0 \\
           0
         \end{psmallmatrix}$} ()
         edge [left]node [scale=0.7]{$ \begin{psmallmatrix}
           1 \\
           \times \\
           0 \\
           1 
          \end{psmallmatrix}$} (q_3)
           edge node[scale=0.7] {$ \begin{psmallmatrix}
           0 \\
           \times\\
           1 \\
           1 
          \end{psmallmatrix}$} (q_1)
          (q_1) edge [loop above] node [scale=0.85]
    {$ \begin{psmallmatrix}
           0 \\
           \times \\
           0 \\
           0
         \end{psmallmatrix},\begin{psmallmatrix}
           0 \\
           \times \\
           1 \\
           1
         \end{psmallmatrix}$} ()
         edge node [scale=0.7]{$ \begin{psmallmatrix}
           1 \\
           \times \\
           0 \\
           1 
          \end{psmallmatrix}$} (q_2)
          (q_2) edge [loop below] node [scale=0.85]
    {$ \begin{psmallmatrix}
           0 \\
           \times \\
           0 \\
           0
         \end{psmallmatrix},\begin{psmallmatrix}
           1 \\
           \times \\
           0 \\
           1
         \end{psmallmatrix},\begin{psmallmatrix}
           0 \\
           \times \\
           1 \\
           1
         \end{psmallmatrix}$} ()
         (q_3) edge [loop below] node [scale=0.85]
    {$ \begin{psmallmatrix}
           0 \\
           \times \\
           0 \\
           0
         \end{psmallmatrix},\begin{psmallmatrix}
           1 \\
           \times \\
           0 \\
           1
         \end{psmallmatrix}$} ()
         edge node [scale=0.7]
    {$\begin{psmallmatrix}
           0 \\
           \times\\
           1 \\
           1 
          \end{psmallmatrix}$} (q_2);
 \end{tikzpicture}
 \caption{Automata for Example 3.1, Part 2} \label{fig:M6}
  \end{subfigure}
\caption{Automata example}
\end{figure}
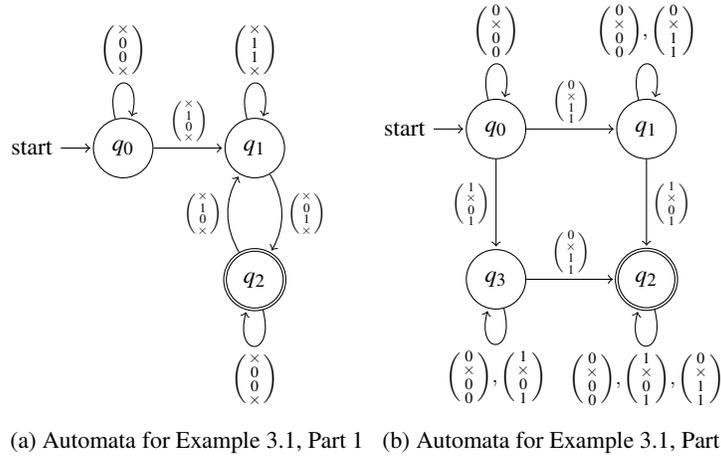

\begin{figure}[h]  
  
\begin{subfigure}[b]{0.4\linewidth}
   \begin{tikzpicture}[>=stealth,shorten >=1pt,auto,node distance=2.5cm]
    \node[state,initial,accepting] (q_0)                {$q_0$};
    \node[state] (q_1) [right of = q_0] {$q_1$};
    \path[->] (q_0) edge [bend left]  node [scale=0.7]{$\begin{psmallmatrix}
           1 \\
           0 \\
           \times\\
            \times
          \end{psmallmatrix}$} (q_1)
          edge [loop above] node [scale=0.85]{$\begin{psmallmatrix}
           0 \\
           0 \\
           \times \\
            \times
          \end{psmallmatrix}$} ()
              (q_1) edge [bend left] node [scale=0.7]{$\begin{psmallmatrix}
              0\\
           1 \\
            \times \\
           \times
          \end{psmallmatrix}$} (q_0)
          edge [loop above] node [scale=0.85]{$\begin{psmallmatrix}
           1 \\
           1 \\
            \times\\
            \times
          \end{psmallmatrix}$} ();
\end{tikzpicture}
\caption{Automata for Example 3.1, Part 3}\label{fig:M7}
\end{subfigure}
\begin{subfigure}[b]{0.5\columnwidth}
   \begin{tikzpicture}[shorten >=1pt,node distance=3.5cm,on grid,auto]
   \node[state,initial] [scale=0.7][align=left](q_0) {$h(x)+b=^?_{\downarrow}x+y$};  
   \node[state] (q_1) [scale=0.8][right=of q_0] {$b=^?_{\downarrow}y$};
   \node[state] (q_3) [scale=0.7][below=of q_0] {$h(x)=^?_{\downarrow}x+y$};
   \node[state,accepting] [scale=0.8](q_2) [below=of q_1] {$0=^?_{\downarrow}0$};
   \path[->]
   (q_0) edge [loop above,scale=0.2] node [align=left] [scale=0.9]{$  \{x \mapsto h(x)+b,$\\$y\mapsto h(y)\}$} ()
         edge[left] node[align=left] [scale=0.9] {$  \{x \mapsto h(x),$\\$y\mapsto h(y)+b\}$} (q_3)
           edge node [align=left] [scale=0.9]{$  \{x \mapsto b,$\\$y\mapsto h(y)\}$} (q_1)
          (q_1) 
         edge [left]node {$ \{y \mapsto b\}$} (q_2)
         (q_3) edge [loop right] node[align=left] [scale=0.9] {$  \{x \mapsto h(x),$\\$y\mapsto h(y)\}$} ();
 \end{tikzpicture}
    \caption{Substitution producing automaton} \label{fig:CE}  
  \end{subfigure}
\caption{Automata example}
\end{figure}
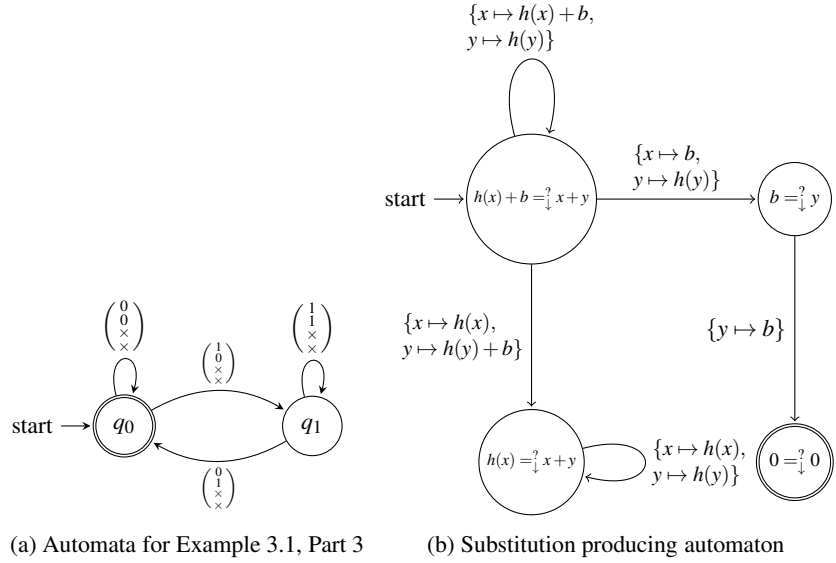
Once we have automata constructed for all the formulas, we take the
intersection and check if there exists a string accepted by all the
automata. If the intersection is not empty, then we have a
solution or an asymmetric unifier for the given problem.

\begin{algorithm}
	\caption{$ACUNh$-decision Procedure for a single constant}
	\label{Alg:ACUNhDec}
	\begin{algorithmic}
		\REQUIRE Asymmetric $ACUNh$-unification problem $S$.
		\STATE For $S$ construct automata for each equation as outlined
		in the paragraph ``Problems with one constant". 
		Let these be $A_1, A_2, \ldots, A_n$.
		\STATE ``Intersect the automata'': Let $\mathcal{A}$ be the automaton that
                recognizes \( \mathop{\bigcap}^{n}_{i=1} L( A_i ) . \)
		\IF{$L( \mathcal{A} ) = \emptyset$}
		\STATE return `no solution.'
		\ELSE
		\STATE return a solution. 
		\ENDIF
	\end{algorithmic}
\end{algorithm}

\paragraph{Problems with more than one constant}: This technique can be extended to the case
where we have more than one constant. 
Suppose we have $k$~constants, say~$c_1 , \ldots , c_k$.
We express each variable~$X$ in terms of the constants as follows: $
X ~ = ~ X_{}^{c_1} + \ldots + X_{}^{c_k} $. For example, if $\mathsf{Y}$
is a variable and $\mathsf{a,b,c}$ are the constants in the
problem, then we create the equation \mbox{$\mathsf{Y}$ = $\mathsf{Y^a+Y^b+Y^c}$.}

If we have an equation $\mathsf{X=h(Y)}$ with constants $\mathsf{a,b,c}$, then we have equations 
$\mathsf{X^a=h(Y^a)}$, 
$\mathsf{X^b=h(Y^b)}$ and $\mathsf{X^c=h(Y^c)}$.
However, if it is an asymmetric equation $\mathsf{X=_{\downarrow}h(Y)}$
all $\mathsf{Y^a}$, $\mathsf{Y^b}$ and $\mathsf{Y^c}$ cannot be zeros simultaneously. \label{LCR}

Similarly, if the equation to be solved is $\mathsf{X = W+Z}$, with $\mathsf{a,b,c}$ as constants,
we form the equations $\mathsf{X^a=W^a+Z^a}$, $\mathsf{X^b=W^b+Z^b}$ and
$\mathsf{X^c=W^c+Z^c}$ and solve the eq-uations. But if it is an
asymmetric equation $\mathsf{X = _{\downarrow}W+Z}$ then we cannot have $\mathsf{W^a}$,
$\mathsf{W^b}$, $\mathsf{W^c}$ to be all zero simultaneously, and similarly
with $\mathsf{Z^a}$, $\mathsf{Z^b}$, $\mathsf{Z^c}$.

Our approach is to design a nondeterministic
algorithm by 
guessing which constant component in each variable has to be~$0$, i.e.,
for each variable~$Y$ and each constant~$b$, we ``flip a coin''
as to whether $Y^b$ will be set equal to~$0$ by the target
solution\footnote{The linear constant restrictions in Section~5 can also be handled this way:
  a constant restriction of the form $a \; \not\in \; X$ can be taken care of by setting $\mathsf{X^a} = 0$}. Now for the case
$\mathsf{X =_{\downarrow} W+Z}$, we do the following:
\begin{quote}
\begin{tabbing}
for \= all constants $a$ do:\\
    \> if $\mathsf{X^a = W^a = Z^a = 0}$ then skip\\
    \> else \= if $\mathsf{W^a = 0}$ then set $\mathsf{X^a = Z^a}$\\
    \>      \> if $\mathsf{Z^a = 0}$ then set $\mathsf{X^a = W^a}$\\
    \>      \> if both $\mathsf{W^a}$ and $\mathsf{Z^a}$ are non-zero then set
$\; \mathsf{X^a =_\downarrow W^a + Z^a}$
\end{tabbing}
\end{quote}

Similarly, for the case $\mathsf{X =_{\downarrow}h(Y)}$ we follow these steps:
\begin{quote}
\begin{tabbing}
for \= all constants $a$ do:\\
    \> if $\mathsf{X^a = Y^a = 0}$ then skip\\
    \> else set \= $\mathsf{X^a =_{\downarrow} h(Y^a)}$
\end{tabbing}
\end{quote}
This is summarized in Algorithm~\ref{Alg:MultiCon}. Thus, it follows that

\begin{algorithm}
	\caption{Nondeterministic Algorithm when we have more than one constant}
	\label{Alg:MultiCon}
	\begin{algorithmic}
	\IF{ there are $m$ variables and $k$ constants}
	\STATE represent each variable in terms of its $k$~constant components.
			\STATE Guess which constant components have to be~$0$.
			\STATE Form symmetric and asymmetric equations for each constant.
			\STATE Solve each set of equations by the Deterministic Finite Automata
			(DFA) construction as outlined in Algorithm~\ref{Alg:ACUNhDec}.  
			\ENDIF
			
	\end{algorithmic}
\end{algorithm} 


\begin{theorem}
  Algorithm~\ref{Alg:MultiCon} is a decision procedure for ground asymmetric unification modulo~$( R_2^{}, \, ACh )$.
\end{theorem}
\begin{proof}
This holds by construction, as outlined in ``Problems with only one constant" and ``Problems with more than one constant".
\end{proof}

We now show that general asymmetric unification modulo~$ACUNh$, where
the solutions need not be ground solutions over the current set of
constants, is decidable by showing that a general solution exists if and only
if there is a \emph{ground} solution in the extended signature where we add
an extra constant.

We represent each term as a sum of terms of the form~$h_{}^i (\alpha)$ where $\alpha$
is either a constant or variable. The superscript (power)~$i$ is referred to as
the \emph{degree} of the simple term~$h_{}^i (\alpha)$. The degree of a term is the
maximum degree of its summands.

\ignore{
\begin{appxlem}~\label{Lemmaone}
  Let $t$ be an irreducible term and $d$ be its degree. Let $\, \mathcal{V}ar(t) \; = \; \{ X_1, X_2 , \ldots , X_n \}$.
  Suppose $c$~is a constant that does not appear in~$t$. Then $ t \theta$ is irreducible, where \[ \theta \; = \; \big\{ X_1 \mapsto c, \; X_2 \mapsto h^{d+1} (c), \; X_3 \mapsto h^{2(d+1)} (c), \; \ldots , \;
  X_n \mapsto h^{(n - 1)(d+1)} (c) \big\}. \]
\end{appxlem}
}

\begin{appxlem}~\label{Lemmatwo}
  Let $t$ be an irreducible term and $d$ be its degree. Let $\, \mathcal{V}ar(t) \; = \; \{ X_1, X_2 ,$ \ldots , $ X_n \}$.
  Suppose $c$~is a constant that does not appear in~$t$. Then
  for any~$D > d$, $ t \theta$ is irreducible,
  where $ \theta \; = \; \big\{ X_1 \mapsto c, \; X_2 \mapsto h^{D} (c), \; X_3 \mapsto h^{2D} (c), \; \ldots , \;
  X_n \mapsto h^{(n - 1)(D)} (c) \big\}. $
\end{appxlem}

\begin{appxlem}~\label{Lemma 13}
Let $\, \Gamma ~ = ~ \big\{ s_1^{} \, {\approx}_{\downarrow}^? \, t_1^{}, \ldots , s_n^{} \, {\approx}_{\downarrow}^? \, t_n^{}
\big\}$ be an asymmetric unification problem. Let $\beta$ be an asymmetric unifier of~$\Gamma$
and~$V = \mathcal{V}\!\mathcal{R} an ( \beta ) = \{ X_1, \ldots X_m \}$. Let
$D = 1 + \mathop{max}\limits_{1 \le i \le n}^{} degree (   s_i^{} \beta, t_i^{} \beta )$, and $c$~be a constant that
does not appear in~$\Gamma$. Then $\theta = \big\{ X_1 \mapsto c, \; X_2 \mapsto h^{D} (c),
\; \ldots , \; X_m \mapsto h^{(n - 1)D} (c) \big\} $ is an asymmetric unifier of~$\; \Gamma$.
\end{appxlem}

\noindent
General solutions \emph{over variables}, without this extra constant~$c$, can be enumerated
by back-substituting (abstracting) terms of the form~$h_{}^j (c)$ and checking whether the
obtained substitutions are indeed solutions to the problem.

\noindent
The exact complexity of this problem is open.



\section{Automaton to find a complete set of unifiers}\label{Sec:Subpro}


In this section we create automata to find all solutions of an ACUNh asymmetric unification problem with constants. 
We also have linear constant restrictions and disequalities for combination.
Our terms will be built from elements in the set described below.  

\begin{definition}
Let $C$ be a set of constants and $X$ be a set of variables.  Define $H(X,C)$ as the set $\{h^i(t) \,\,|\,\, t \in X \cup C\}$.  
We also define $H_n(X,C)$ as $\{h^i(t) \,\,|\,\, t \in X \cup C, i \leq n\}$.
For any object $t$ 
we define 
$Const(t)$ to be the set of constants in $t$, except for 0.
For an object $t$, define $H(t) = H(Var(t),Const(t))$ and  $H_n(t) = H_n(Var(t),Const(t))$. 
\end{definition}


Terms are sums.  We often need to talk about the multiset of terms in a sum.

\begin{definition}
Let $t$ be a term whose $R_h$ normal form is $t_1 + \cdots +t_n$.  Then we define $mset(t) = \{t_1,\cdots,t_n\}$.  Inversely, if $T = \{t_1,\cdots,t_n\}$ then $\Sigma T = t_1 + \cdots + t_n$.
\end{definition}

A term in normal form modulo $R_1$ can be described as a sum in the following way.

\begin{theorem}
Let $t$ be a term in $R_1$ normal form. Then there exists an $H \subseteq H(t)$ such that $t = \Sigma H$.
\end{theorem}

\begin{proof}
Since $t$ is reduced by $h(x+y) \rightarrow h(x) + h(y)$, it cannot have an $h$ symbol above a $+$ symbol.  So it must be a sum of terms of the form $h^i(s)$ where $i \geq 0$ and $s$ is a constant.  Since $t$ is also reduced by $R_2$, there can be no duplicates in the sum.  
\end{proof}

We show that every substitution $\theta$ that is irreducible with respect to $R_1$, can be represented as a sequence of smaller substitutions, which we will later use to construct an automaton. 


\begin{definition}
Let $\zeta$ be a substitution and $X$ be a set of variables.  Then $\zeta$ is a {\em zero substitution on $X$} if $Dom(\zeta) \subseteq X$ and $x\zeta = 0$ for all $x \in Dom(\zeta)$.
\end{definition}

\begin{theorem}

Let $t$ be an object and
$\theta$ be a ground substitution in $R_1$ normal form, such that $Dom(\theta) = Var(t)$.
Let $m$ be the maximum degree in $mset(Ran(\theta))$
Then there are substitutions $\zeta$, $\theta_0,\cdots,\theta_m$ 
such that
\begin{enumerate}
    \item $\zeta$ is a zero substitution on $Dom(\theta)$,
    \item $\zeta\theta_0 \cdots \theta_m = \theta$,
    \item $Dom(\theta_i) = Var(t\zeta\theta_0 \cdots \theta_{i-i})$
    \item for all $i$ and all variables $x$ in 
    $Dom(\theta_i)$,
    $x\theta_i = \Sigma T$ for some nonempty $T \subseteq Const($ Ran($\theta)) \cup \{h(x)\}$.
\end{enumerate}
\end{theorem}

\begin{proof}
By the previous theorem, we know that each $x\theta$ is a sum of $h$-terms or is 0.  Then $\zeta$ and $\theta_i$ can be defined as follows, where $S^x = mset(x\theta)$ and ${S_i}^x$ is the set of terms in $S$ with degree $i$:
\begin{itemize}
\item If $x\theta = 0$ then $x\zeta = 0$ else $x\zeta = x$.
\item For all $x \in Dom(\theta_i)$
\begin{itemize}
\item If the maximum degree of $S^x$ is $i$ then $x\theta_i = \Sigma {S_i}^x$.
\item If no terms in $S^x$ have degree $i$ then $x\theta_i = h(x)$.
\item If $S^x$ has terms of degree $i$ and also terms of degree greater than $i$ then $x\theta_i = h(x) + \Sigma {S_i}^x$.
\end{itemize}
\end{itemize}
\end{proof}

In the rest of this section we will be considering the ACUNh asymmetric equation $u \asymeq v$, where $u$ and $v$ are in $R_1$ normal form, and we will build an automaton to represent all the solutions of $u \asymeq v$.  We will need the following definitions.

\begin{definition}
Let $t$ be an object. Define $loseh(t) = \Sigma \{h^i(t) \,\,|\,\, h^{i+1}(t) \in mset(t\downarrow_{R_h})\}$.
\end{definition}

In the next four automata definitions we will use the following notation:
Let $P$ be a set of ACUNh asymmetric equations.
Let $m$ be the maximum degree of terms in $P$.
Let $\Theta$ be the set of all substitutions $\theta$ such that $Dom(\theta) \subseteq Var(P)$ and for all $x \in Dom (\theta)$, $x\theta = \Sigma T$ where $T$ is a nonempty subset of $Const(P) \cup \{h(x)\}$. 
Let $u \asymeq v$ be an ACUNh asymmetric equation.

First we define an automaton to solve the $ACUNh$ asymmetric unification problem with constants.

\begin{definition}
The automaton $M(u \asymeq v,P)$ consists of the quintuple 
$(Q,q_{u \asymeq v},F,$ $\Theta,$ $\delta)$, 
where $Q$ is the set of states, $q_{u \asymeq v}$ is the start state, $F$ is the set of accepting states, $\Theta$ is the alphabet, and $\delta$ is the transition function, defined as follows:
\begin{itemize}
    \item $Q$ is a set of states of the form $q_{s \asymeq t}$, where $s = \Sigma S$ and $t = \Sigma T$, for some $S$ and $T$ subsets of $H_m(P)$.
    \item $F = \{q_{s \asymeq t} \in Q \,\,|\,\, mset(s) = mset(t)\}$
    \item $\delta : Q \times \Theta \longrightarrow Q$ such that $\delta(q_{s \asymeq t}, \theta) = q_{loseh(s\theta)\downarrow_{R_1} \asymeq loseh(t\theta)}$ if 
    $Dom(\theta_i) = Var$ ($s \asymeq t)$,
    $mset((s\theta)\downarrow_{R_1}) \cap H_0(P) = mset(t\theta) \cap H_0(P)$, and
    $mset(t\theta)$ contains no duplicates.
    \ignore{
    \begin{enumerate}
        \item $Dom(\theta_i) = Var(s \asymeq t)$,
        \item $mset((s\theta)\downarrow_{R_1}) \cap H_0(P) = mset(t\theta) \cap H_0(P)$, and
        \item $mset(t\theta)$ contains no duplicates.
    \end{enumerate}
}
\end{itemize}
\end{definition}

Next we define an automaton to solve linear constant restrictions.

\begin{definition}
Let $R$ be a set of linear constant restrictions of the form $(x,c)$.
$M_{LCR}(R,$ $P) = (\{q_0\},$ $q_0, \{q_0\}, \Theta,\delta_{LCR})$ where
$\delta_{LCR}(q_0,\theta) = q_0$ if for all variables $x$ and all $(x,c) \in R$, $c \not\in Const(x\theta)$.
\end{definition}

Next we define an automaton to solve disequalities between a variable and a constant.

\begin{definition}
Let $D$ be a set of disequalities of the form $x \not= c$ where $x$ is a variable and $c$ is a constant.
$M_{VC}(D,P) = (\{q_0,q_1\}, q_0, $ $\{q_0,q_1\}, \Theta,\delta_{VC})$ where $\delta_{VC}(q_0,\theta) = q_1$ if for all variables $x$ and all $x \not= c \in D$,
$x\theta \not= c$.  Also $\delta_{VC}(q_1,\theta) = q_1$.
\end{definition}

Finally we define automata for solving disequalities between variables

\begin{definition}
Let $x$ and $y$ be variables.  Then $M_{VV}(x \not= y,P) = (\{q_0,q_1\}, q_o, \{q_1\}, \Theta,$ $\delta_{x \not= y})$ where $\delta_{x \not= y}(q_0,\theta) = q_0$ if $mset(x\theta) = T \cup \{h(x)\}$ and $mset(y\theta) = T \cup \{h(y)\}$ for some $T$.  Also $\delta_{x \not= y}(q_0,\theta) = q_1$ if $mset(x\theta) \not= mset(y\theta)$ and $mset(x\theta)[x \mapsto y] \not= mset(y\theta)$.


\end{definition}

These are all valid automata. In particular, the first automaton described has a finite number of states, and each transition yields a state in the automaton.
Now we show that these automata can be used to find all asymmetric ACUNh unifiers.

We need a few properties before we show our main theorem, that the constructed automaton finds all solutions.

\begin{appxlem}
\label{novar}
Let $t$ be an object and $\theta$ be a substitution, such that, for all $x \in Var(t)$, $mset(x\theta)$ does not contain a variable. Then $mset(t\theta)$ does not contain a variable.   
\end{appxlem}

\begin{proof}
Consider $s \in mset(t)$.  If $s$ is not a variable then $s\theta$ is not a variable. If $s$ is a variable, then, by our hypothesis, $s\theta$ is not a variable. 
\end{proof}

\begin{appxlem}
\label{failunif}
Let $s \asymeq t$ be an ACUNh asymmetric unification equation in $P$, where $mset(s)$ and $mset(t)$ contain no variables and $mset(s\downarrow_{R_1}) \cap H_0(P) \not= mset(t) \cap H_0(P)$. Then for all substitutions $\sigma$, $s\sigma$ and $t\sigma$ are not unifiable.    
\end{appxlem}

\begin{proof}
$s$ and $t$ are not unifiable, because, wlog,  the multiplicity of some constant in $mset(s\downarrow_{R_1})$ is not in $mset(t\downarrow_{R_1})$. When we apply a substitution, that same constant will appear in $mset((s\sigma)\downarrow_{R_1})$ but not $mset((t\sigma)\downarrow_{R_1})$, since $mset(s)$ and $mset(t)$ contain no variables. So $s\sigma$ and $t\sigma$ are not unifiable.
\end{proof}

\begin{appxlem}
\label{failasym}
Let $t$ be such that $mset(t)$ contains a duplicate. Then $ \forall ~\sigma$, $t\sigma$ is reducible by $R_2$.  
\end{appxlem}

\begin{proof}
We know $t$ is reducible by $R_2$ because $mset(t)$ contains a duplicate.  But then $t\sigma$ also contains a duplicate.  
\end{proof}

\begin{appxlem}
\label{succeed}
Let $s \asymeq t$ be an ACUNh asymmetric unification equation in $P$, such that $mset(s)$ and $mset(t)$ contain no variables. Suppose also that $mset(s\downarrow_{R_1}) \cap H_0(P) = mset(t) \cap H_0(P))$ and $mset(t)$ contains no duplicates. Then $\sigma$ is an ACUNh asymmetric unifier of $s \asymeq t$ if and only if $\sigma$ is an ACUNh asymmetric unifier of 
$loseh(s\downarrow_{R_1}) \asymeq loseh(t)$.
\end{appxlem}

\begin{proof}
Let $s' = loseh(s\downarrow_{R_1})$ and $t' = loseh(t)$. If $mset(s')$ and $mset(t')$ contain no constants, then $s \asymeq t$ and $s' \asymeq t'$ have the same solutions.
Since $mset(s)$ and $mset(t)$ contain no variables, the multiset of constants in $s$ is the same as the multiset of constants in $s\sigma$.  Similarly for $t$ and $t\sigma$. Therefore $s \asymeq t$ has the same solutions as $s' \asymeq t'$.  
\end{proof}

\begin{theorem}
Let $P$ be a set of asymmetric ACUNh equations, such that all terms in $P$ are reduced by $R_1$. Let $\theta$ be a substitution which is reduced by $R_1$. Let $R$ be a set of linear constant restrictions. Let $D$ be a set of variable/constant disequalities. Let $D'$ be a set of variable/variable disequalities.

Then 
$\theta$ is a solution to $P$ if and only if there exists a zero substitution $\zeta$ on $P$ where all right hand sides in $P$ are irreducible, and a sequence of substitutions $\theta_0, \cdots, \theta_m$ such that $\theta \leq \zeta\theta_0 \cdots \theta_m$ and 
\begin{enumerate}
\item
The string $\theta_0 \cdots \theta_m$ is accepted by $M((u\asymeq v)\zeta)\downarrow_{R_1},P'\zeta)$ for all $u \asymeq v \in P$.
\item 
The string $\theta_0 \cdots \theta_m$ is accepted by
$M_{LCR}(R,P'\zeta)$.
\item 
The string $\theta_0 \cdots \theta_m$ is accepted by 
$M_{VC}(D,P'\zeta)$.
\item 
The string $\theta_0 \cdots \theta_m$ is accepted by 
$M_{VV}(x \not= y,P'\zeta)$ for all $x \not= y \in D'$.
\end{enumerate}
where $P' = P \cup \{c \asymeq c\}$ for a fresh constant $c$.
\end{theorem}

\begin{proof}
First we show that Item 1 holds for a ground substitution $\theta$ reduced by $R_1$.
By the previous theorem, $\theta$ can be represented as $\zeta\theta_0 \cdots \theta_m$.  

We show by induction that, for all $i$, if $\theta = \zeta\theta_0 \cdots \theta_i$ and $\delta(q_{(u \asymeq v)\zeta},\theta_0 \cdots \theta_i) = q_{s \asymeq t}$
then $\theta\sigma$ is an asymmetric ACUNh unifier of $u \asymeq v$ if and only if $\sigma$ is an asymmetric ACUNh unifier of $s \asymeq t$.  In the base case, $\theta = \zeta$ and $(s \asymeq t) = (u \asymeq v)\zeta$, so it is true.  

For the inductive step, we assume the statement is true for $i$ and prove it for $i+1$.  
Then let $\sigma'$ be an arbitrary substitution, and instantiate $\sigma$  $\theta\sigma'$ in the inductive assumption, where $\theta = \zeta\theta_0 \cdots \theta_i$.
Our assumption implies that $\theta\theta_{i+1}\sigma$ is an asymmetric ACUNh unifier of $u \asymeq v$ if and only if $\theta_i\sigma$ is an asymmetric ACUNh unifier of $s \asymeq t$
(i.e., $\sigma$ is an asymmetric ACUNh unifier of $(s \asymeq t)\theta_{i+1}$).
If we can now show that $\sigma$ is an asymmetric ACUNh unifier of $(s \asymeq t)\theta_{i+1}$ if and only if $\sigma$ is an ACUNh unifier of $loseh(s\theta_{i+1}) \asymeq loseh(t\theta_{i+1})$. and $mset(s\theta_{i+1} \cap H_0(P') = mset(t\theta_{i+1}) \cap H_0(P')$ and $mset(t\theta_{i+1})$ contains no duplicates, then we are done. 
By Lemma~\ref{novar}, we know that $mset(s \asymeq t)\theta_{i+1}$ contains no variables.  Then we apply Lemma~\ref{succeed} to prove the induction step.


This proves our inductive statement.  
If $\theta$ is not an asymmetric ACUNh unifier of $u \asymeq v$, then Lemmas~\ref{failunif} and \ref{failasym} imply that the transition function will not be applicable at some point.
Our inductive statement shows that 
$\theta$ is an asymmetric ACUNh unifier of $u \asymeq v$ if and only if there is a final state with $id$ as an asymmetric ACUNh unifier, which will be an accepting state. 

This concludes the case for a ground substitution $\theta$.  It $\theta$ is not ground, then the fact that $P'$ contains a fresh constant $c$ means that we create substitutions with an additional constant.  We have already shown in this paper that nonground solutions are generalizations of solutions involving one additional constant. 

It is straightforward to see that the other automata only accept valid solutions of linear constant restrictions and disequations.
\end{proof}

If desired, we could intersect all the automata, yielding an automaton representing all the solutions of the problem (think of the results after applying $\zeta$ as a set of initial states).  This shows that the set of solutions can be represented by a regular language, with or without $LCRs$ 
 and disequalities.  If we only want to decide asymmetric unification, we just check if there is an accepting state reachable from an initial state.  We could enumerate all the solutions by finding all accepting states reachable in 1 step, 2 steps, etc.  If there is a cycle on a path to an accepting state, then there are an infinite number of solutions, otherwise there are only a finite number of solutions.
This will find all the ground substitutions. 
To find all solutions, we generalize the solutions we find and check them. Indeed, the only terms that need to be generalized are those containing $c$.  This is decidable because there are only a finite number of generalizations.

In Figure~(\ref{fig:CE}), we show the automaton created for the problem $h(x) + b \asymeq x + y$, without linear constant restrictions and disequality constraints. In this example, the only zero substitution that works is the identity.  Notice that $c$ never appears in the domain of a substitution, because no such substitution satisfies the conditions for the transition function.  This leads to the following theorem.

\begin{theorem}\label{Thm:not_finitary}
Asymmetric ACUNh unification with constants is not finitary.
\end{theorem}

\begin{proof}
The automaton constructed for $h(x) + b \asymeq x + y$ has a cycle on a path to an accepting state. Therefore there are an infinite number of solutions.  Since there is no $c$ in the range of the solution, all the solutions are ground.  So no solution can be more general than another one, which means this infinite set of solutions is also a minimal complete set of solutions.
\end{proof}

\section{Combining Automata based Asymmetric Algorithms with the Free Theory}~\label{Sec:Combination}
In order to obtain a general asymmetric ACUNh-unification decision
procedure we need to add free function symbols. We can do this by
using disjoint combination. The problem of asymmetric unification in
the combination of disjoint theories was studied
in~\cite{DBLP:conf/fossacs/ErbaturKMMNR14} where an algorithm is
developed for the problem. However, the algorithm
of~\cite{DBLP:conf/fossacs/ErbaturKMMNR14} does not immediately apply
to the two methods developed in this paper. This is due to the nature
of the two automata based approaches. More formally, let $\Delta_1$
and $\Delta_2$ denote two equational theories with disjoint signatures
$\Sigma_1$ and $\Sigma_2$. Let $\Delta$ be the combination, $\Delta =
\Delta_1 \cup \Delta_2$, of the two theories having signature
$\Sigma_1 \cup \Sigma_2$.  The algorithm
of~\cite{DBLP:conf/fossacs/ErbaturKMMNR14} solves the asymmetric
$\Delta$-unification problem. It assumes that there exists a finitary
complete asymmetric $\Delta_i$-unification algorithm with linear
constant restrictions, $A_i$. Based on this assumption the algorithm
is able to check solutions produced by the $A_1$ and $A_2$ algorithms
for theory-preserving and injective properties, discarding those that
are not. A substitution $\sigma_i$  is \emph{injective} modulo $\Delta_i$ if $x \sigma_i =_{\Delta_i} y \sigma_i$ iff $x = y$,
and $\sigma_i$ is \emph{theory preserving} if for any variable $x$ of index $i$, $x \sigma_i$ is not  a variable of index $j \neq i$.
For the automaton it is not always possible to check
solutions, however, it is possible to build constraints
into the automaton that enforce these conditions.  Algorithm~\ref{Alg:comb} is a modification
of the algorithm from~\cite{DBLP:conf/fossacs/ErbaturKMMNR14} with the
following properties:
\begin{itemize}
	\item $\Delta_1 = ACUNh$ and $\Delta_2= F_{\Omega}$, for
	some free theory $F_{\Omega}$ with symbols $\Omega$.
	\item For each $\Delta_1$-pure problem, partition, and theory index, an
	automaton is constructed enforcing the injective and theory preserving 
	restrictions. Since these restrictions are built into the automata, the only $\Delta_1$ solutions produced will be both theory preserving and injective. 
	\item The solution produced by the $F_{\Omega}$ algorithm is checked as in the original algorithm. If the solution is found not to be injective or theory preserving it is discarded.
\end{itemize}
The new modified version is presented in Algorithm~\ref{Alg:comb} (included in the appendix due to space). 
Given the decision 
procedure of Section~\ref{Sec:Decpro} we obtain the following.

\begin{theorem}
Assume there exists an asymmetric $ACUNh$  decision procedure 
that enforces linear constant restrictions, theory indexes, and injectivity. Then 
Algorithm~\ref{Alg:comb} is a general asymmetric $ACUNh$  decision procedure.
\end{theorem}
\begin{proof}
	The result follows directly from the proof contained in~\cite{DBLP:conf/fossacs/ErbaturKMMNR14}. There it is shown
	that Algorithm~\ref{Alg:comb} is both sound and complete.
	The only modification is that in~\cite{DBLP:conf/fossacs/ErbaturKMMNR14} the combination
	algorithm checks the $\Delta_1$ solutions for the properties
	of being injective and theory preserving, while in 
	Algorithm~\ref{Alg:comb} it is assumed that the algorithm $A_1$ itself will enforce these restrictions. 
\end{proof}

If instead of a decision procedure we want to obtain a general 
asymmetric $ACUNh$ unification algorithm we can 
use the automata based algorithm from Section~\ref{Sec:Subpro}
and again a modification of the asymmetric combination algorithm 
of~\cite{DBLP:conf/fossacs/ErbaturKMMNR14}. Here, the
modification to the combination algorithm is even smaller. We just remove the 
check on injective and theory preserving substitutions. Again these restrictions are enforced by the automata. 
The solutions to the $ACUNh$ and the free theory are combined as is done in~\cite{DBLP:conf/fossacs/ErbaturKMMNR14} since they obey the same linear constant restrictions. Since asymmetric $ACUNh$ unification with constants is not finitary (Theorem~\ref{Thm:not_finitary}), the
general asymmetric $ACUNh$ unification algorithm will not 
in general produce a finite set of solutions. However, 
based on the algorithm of Section~\ref{Sec:Subpro} we easily
obtain the following result.
    
\begin{theorem}
	Assume there exists an asymmetric $ACUNh$  algorithm that enforces linear constant restrictions, theory indexes and injectivity, and produces a complete set of unifiers. Then there exists a general asymmetric $ACUNh$  
	algorithm producing a complete set of unifiers.
\end{theorem}

\section{Conclusion}\label{sec:conclusion}

We have provided a decision procedure and an  algorithm for asymmetric unification  modulo $ACUNh$ using a decomposition $R \uplus ACh$. This is the first example of an asymmetric unification algorithm for a theory in which unification modulo the set $\Delta$ of axioms is undecidable.  It also has some practical advantages: it is possible to tell by inspection of the automaton used to construct unifiers whether or not a problem has a finitary solution. Moreover, the construction of the automaton gives us a natural way of enumerating solutions; simply traverse one of the loops one more time to get the next unifier. 

There are a number of ways in which we could extend this work.  For example, the logical next step is to consider the decidability of asymmetric unification of $AGh$ with a $\Delta =  ACh$.   If the methods we used for $ACUNh$ extend to $AGh$, then we have an asymmetric unification algorithm for $AGh$, although with $\Delta = ACh$ instead of $AC$.  On the other hand, if we can prove   undecidability of asymmetric unification for $AGh$ with $\Delta = ACh$  as well as with $\Delta = AC$, this could give us new understanding of the problem that might allow us to obtain more general results.  Either way, we expect the results to give increased understanding of asymmetric unification when homomorphic encryption is involved.

\bibliography{citations}

\appendix

\section{Proof of $R_2, ACh^{}$ convergence}
\label{appendix1}
\subsection{$R_2^{}$ is $ACh$-convergent}
To begin with let us show that the theory,  $\rightarrow_{R_2^{} / ACh}^{}$, is terminating.
\begin{appxlem}
  $\rightarrow_{R_2^{} / ACh}^{}$ is terminating.
\end{appxlem}

\begin{proof}
  We use a polynomial interpretation. The signature of $R_2$
  is~$\Sigma \; = \; \{+, h, 0\}$. We take the polynomial
  interpretation $P_h \; = \; 2 * X, \; P_0 \; = \; 1, \; P_+ \; = \;
  X+Y$. It is not hard to see that the rewrite rules are decreasing
  and the identities are preserved under this interpretation. Hence
  $\rightarrow_{R_2, ACh}^{}$~is terminating.
\end{proof}

\noindent
The identities in $ACh$ can be further decomposed,
retaining {\textit{associativity}} and {\textit{commutativity}} as identities ${\Delta}^{}$ and
viewing the other identity as a rewrite rule $R_h^{}$.

\noindent
Thus $R_h^{}$ is the term rewriting system: 
$h(x+y)  \rightarrow  h(x) + h(y)$
with the identities~$AC$:
  $(x + y) + z  \approx  x + (y + z), \;\;\;
  x + y  \approx  y + x$
Note that $R_1^{} ~ = ~ R_2^{} \cup R_h^{}$. Thus
$\longrightarrow_{R_1^{},AC}^{} \; = \; \longrightarrow_{R_2^{},AC}^{} \; \cup \; \longrightarrow_{R_h^{},AC}^{}$.
All these term rewriting
systems, $R_1^{}$, $R_2^{}$ and $R_h^{}$, are $AC$-convergent.

\noindent
We define some terms which will be used later:
\begin{definition}
  A term is a $\bf{+}$-term if and only if it has only variables and constants, and no occurrences of~$h$.
\end{definition}

\begin{definition}
  A term is a \emph{pure} $\bf{+}$-term if and only if it is a $\bf{+}$-term with no constants
  (i.e., it belongs to $T(\{ + \}, V)$).
\end{definition}

\noindent
In order to prove that $R_2^{}$ is $ACh$-convergent we first prove the
following lemmas:

\begin{appxlem}
If $s$ and $t$ are irreducible by~$\rightarrow_{R_h, AC}^{}$,
then so is $s+t$.
\end{appxlem}
\begin{proof}
The only rule in $R_h^{}$ has $h$ as the root of its left-hand
side. Hence for the term $s+t$ there is no reduction possible.
\end{proof}
\vspace{0.1in}

\begin{corol} \label{normalform}
$(s + t)\!\big\downarrow_{R_h^{}, AC} ~ = ~ s\!\big\downarrow_{R_h^{}, AC} \; + \; t\!\big\downarrow_{R_h^{}, AC}$
\end{corol}

\vspace{0.1in}
\begin{appxlem}\label{lem3}
If s is a $\bf{+}$-term, then $s \; \lesssim^{}_{ACh} \; t$
if and only if $s \; \lesssim_{AC}^{} \; t\!\big\downarrow_{R_h^{},AC}$.
\end{appxlem}
\begin{proof}
Since $s$ is a $\bf{+}$-term, $s$ is of the form $Y_1^{} + \ldots + Y_n^{}+c_1 + \ldots + c_m$,
where the $Y_i^{}$'s need not be distinct. Let $t^{\prime}$ be
the normal form of~$t$ modulo~$R_h^{}$, i.e., $t^{\prime} = t\!\big\downarrow_{R_h^{}, AC}^{}$.\\

\noindent
``if'' part: If $s$ $AC$ matches with $t^{\prime}$, then $s$ is $ACh$ matchable with $t$. \\

Let $\theta$ be a matching substitution such that $s\theta \;
{\approx}_{AC} \; t_{}^\prime$.  The substitution $\theta$ also $ACh$
matches $s$ with~$t$, i.e., $s\theta \; \approx_{ACh} \; t$, because 
$AC$ is a subset of $ACh$.

\noindent
``only if'' part: If $s$ $ACh$ matches with $t$, then $s\;AC$ matches with $t_{}^\prime$. \\

Suppose $s\sigma \; {\approx}_{ACh} \; t$, where $\sigma$ is substitution.
Let $\sigma_{}^\prime$ be a substitution defined as $z\sigma_{}^\prime ~ = ~
z\sigma\!\big\downarrow_{R_h^{}, AC}^{}$ for all $z \in Dom(\sigma)$.
We now show that $s\sigma_{}^\prime\; \approx_{AC} \; t^{\prime}$.

We have $s\sigma \; {\approx}_{ACh} \; t$ and by Corollary~\ref{normalform},
$s\sigma^{\prime} \; = \; s\sigma\!\big\downarrow_{R_h^{}, AC}^{}$, i.e., 
\mbox{$s\sigma \rightarrow_{R_h, AC}^{!}s\sigma^{\prime}$}. Thus
$s\sigma^{\prime} \; \approx_{AC}^{} \; t^{\prime}$.
\end{proof}

\begin{corol} \label{R2ACh}
A term $s$ is reducible at the root (i.e., at position~$\epsilon$) by
$\; \longrightarrow_{R_2^{},ACh}^{} \; $
if and only if 
$s\!\big\downarrow_{R_h^{}, AC}^{}$ is reducible at the root 
by $\; \longrightarrow_{R_2^{},AC}^{}$.
\end{corol}

\begin{appxlem}
If $s$ is a pure $\bf{+}$-term, then $s \; \lesssim^{}_{ACh} \; h(t)$
if and only if $s \; \lesssim_{ACh}^{} \; t$.
\end{appxlem}
\begin{proof}
For the ``if'' part, let
$\theta$ be the matching substitution. Then the substitution
$\theta_{}^\prime$, defined as $ y\theta_{}^\prime = h( y\theta )$ $\text{ for all } y$ matches
$s$ with $h(t)$.

For the ``only if'' part, suppose $\beta$ is a substitution that
$ACh$-matches $s$ with~$h(t)$. We can assume without loss of
generality that $t$ is a ground term. Since $R_h^{}$ is
$AC$-convergent, we can also assume that the normal form of~$h(t)$,
modulo~$\rightarrow_{R_h^{}, \, AC}^{}$, is of the form~$h_{}^{i_1}
(a_1) + \ldots ~ + h_{}^{i_n} (a_n)$ where the $a_i^{}$'s are constants
and each exponent~$i_j^{}$ is a positive integer greater than zero.
Therefore for each variable~$X_k^{}$, $X_k^{}\beta ~ \approx_{ACh}^{} ~ h(t_k^{})$
for some ground term~$t_k^{}$. Thus $\beta_{}^\prime$,
defined as $ y\beta _{}^\prime = s_{}^\prime ~ ~ \text{ if } ~ ~ y\beta~ \approx_{ACh}^{} ~ h( s_{}^\prime )$ will
match $s$ with~$t$.
\end{proof}

\begin{appxlem}\label{lem5}
  If $s$ is a $\bf{+}$-term and $s \; \lesssim^{}_{ACh} \; h(t)$
  then $s \; \lesssim_{ACh}^{} \; t$.
\end{appxlem}

\begin{proof}
  Since $s$ is a $\bf{+}$-term, $s$ must be of the form
  $X_1^{} + \ldots + X_m^{}+b_1 + \ldots + b_n$ where some of the $X_i^{}$'s can be equal
  (i.e., variables can be repeated). However, clearly $s$~cannot have constants
  and thus must be a pure~$\bf{+}$-term.
  The rest of the proof follows from the ``only if'' part
  of the previous lemma.
\end{proof}

\begin{appxlem} \label{redexes}
  For all $t, \, s, \, t_{}^{\prime\prime}$: if $t$
  is reducible at the root (i.e., at position~$\epsilon$) by
  $\; \longrightarrow_{R_2^{},ACh}^{} \; $ 
  and $t_{}^{\prime\prime} \; \approx_{AC}^{} \; t \;+\;s$
  then $t_{}^{\prime\prime}$ is also reducible 
  by $\; \longrightarrow_{R_2^{},ACh}^{}$.
\end{appxlem}
\begin{proof}
  Suppose $t\;\approx_{ACh}\;l\sigma$ where $\sigma$ is a substitution
  and $l\rightarrow r\in R_2^{}$. The case where $l = h(0)$ is trivial.

  Suppose $l$ is the left-hand side of one of the other rules. Then $l$~is
  a $+$-term. Let $t_1=t\!\big\downarrow_{R_h^{}, AC}^{}$. By the proof of
  Lemma~\ref{lem3}, $t_1\; \approx_{AC}^{} \;l\sigma_{}^{\prime}$, where
  $\sigma_{}^{\prime}$ is defined as $ y\sigma_{}^\prime  = 
y\sigma\!\big\downarrow_{R_h^{}, AC}^{} $ for all $y \in Dom(\sigma)$.
From this we have $l\sigma^{\prime} \; = \; l\sigma\!\big\downarrow_{R_h^{}, AC}^{}$, i.e., 
\mbox{$l\sigma \rightarrow_{R_h, AC}^{!}l\sigma^{\prime}$}.

  Let $t_{}^{\prime\prime}\; \approx_{AC}^{} \;t\;+\;s$. By Corollary~\ref{normalform},
  $t\;+\;s\;\longrightarrow^{!}_{R_h^{},AC}\;t_1 + \widehat{s} \,$ where
  $\widehat{s} \, = \, s\!\big\downarrow_{R_h^{},AC}^{}$. Since $t_1$ is reducible
  at the root by~$\longrightarrow_{R_2^{},AC}^{}$, so is
  $t_1 + \widehat{s}$, as can be seen from an examination of the rules.
  Thus $ t_{}^{\prime\prime} \; \longrightarrow^{!}_{R_h^{},AC}\;
  t_1 + \widehat{s} \; ~  \; \approx_{AC}^{} \; ~ l_{}^{\prime}\beta $ for some rule~$( l_{}^{\prime}
  \rightarrow r_{}^{\prime} )$~in~$R_2^{}$.
  This implies that $\, t_{}^{\prime\prime} \;   ~  \; \approx_{ACh}^{} \; ~ l_{}^{\prime}\beta $.
\end{proof}

\begin{appxlem} \label{h-terms}
For all $t$, if $h(t)$ is reducible by~$\;
\longrightarrow_{R_2^{},ACh}^{} \;$ then $t = 0$
or $t$ is reducible by~$\;
\longrightarrow_{R_2^{},ACh}^{}$.
\end{appxlem}
\begin{proof}
  Suppose $t \neq 0$ and that $t$ is
  not reducible by $\; \longrightarrow_{R_2^{},ACh}^{}$. Then 
  there must be some left-hand side $l$~in~$R_2^{}$ such that
  $l \; \lesssim^{}_{ACh} \; h(t)$. Since $t \neq 0$, $l$ cannot be~$h(0)$, so $l$ must be a $\bf{+}$-term. 
  If  $l \; \lesssim^{}_{ACh} \; h(t)$, by Lemma~\ref{lem5} we know that $l \; \lesssim^{}_{ACh} \; t$, which is a contradiction. 
\end{proof}


\begin{appxlem}
For all $s, s_{}^{\prime}$, if $s$ is reducible by $\;
\longrightarrow_{R_2^{},ACh}^{} \;$ and $s \, \approx_{AC}^{} \, s_{}^{\prime}$,
then $s_{}^{\prime}$ is reducible by~$\;
\longrightarrow_{R_2^{},ACh}^{}$.
\end{appxlem}

\begin{proof}
  Assume the contrary. Let $t$ be the \emph{smallest} term such that
  $t$~is reducible by $\; \longrightarrow_{R_2^{},ACh}^{} \;$ and
  there is another term $t_{}^{\prime}$, \emph{AC}-equivalent to~$t$,
  such that $t_{}^{\prime}$ is irreducible by $\;
  \longrightarrow_{R_2^{},ACh}^{}$.  Clearly $t$ cannot be reducible
  at the root. Let $p \neq \epsilon$ be an outermost position in~$t$
  such that $t|_p^{}$ is a redex. Let $p = p_{}^{\prime} \cdot i$. 
  Obviously, $t|_p^{}$ cannot
  be~$h(0)$.  The parent of node~$p$ cannot have a $+$-symbol because
  then that subterm, i.e., $t|_{p_{}^{\prime}}^{}$, will be a redex too, by the proof of
  Lemma~\ref{redexes}. Thus it has to be an~$h$: in other words,
  $h( t|_p^{} )$ is a subterm of~$t$. Since $t_{}^{\prime} \, \approx_{AC}^{} \, t$,
  $t_{}^{\prime}$ should have a subterm~$h( t_{1}^{\prime} )$ where
  $t_{1}^{\prime} \, \approx_{AC}^{} \, t|_p^{}$. If $t_{1}^{\prime}$ is
  reducible then so is $t_{}^{\prime}$; if not, then $t_{}^{\prime}$
  is a smaller counterexample.
\end{proof}

\begin{corol} \label{plus-terms}
For all $t, \, t_{}^\prime, \, t_{}^{\prime\prime}$: if $t$ is reducible by $\;
\longrightarrow_{R_2^{},ACh}^{} \;$ and $t_{}^{\prime\prime} \; \approx_{AC} \; t \;+\;t_{}^{\prime}$
then $t_{}^{\prime\prime}$ is also reducible by $\;
\longrightarrow_{R_2^{},ACh}^{}$.
\end{corol}

\begin{appxlem} \label{commute}
For all $s_1, \, s_2, \, s_3$: if $s_1$ $\;
\longrightarrow_{R_h^{},AC}^{} \;$ $s_2$, and $s_2$ $\;
\longrightarrow_{R_2^{},ACh}^{} \;$ $s_3$, then $s_1$ is also reducible
by~$\longrightarrow_{R_2^{},ACh}^{}$.
\end{appxlem}
\begin{proof}
Suppose $s_1$ reduces to $s_2$ at some position $p$ and $s_2$ reduces
to $s_3$ at some position $q$, i.e.,
$s_1 \; \longrightarrow_{R_h^{},AC}^{p} \; s_2$ and
$s_2 \; \longrightarrow_{R_2^{},ACh}^{q} \; s_3$.
Thus there is a substitution $\theta$ such that ${s_1}|_p^{} \; \approx_{AC}^{} \;
(h(x + y))\theta $ and $s_2 \; = \; s_1 [ h(x\theta) + h(y\theta) ]_p^{}$.

$\emph{Two cases}$ have to be considered:

$p < q:$ Either $h(x\theta)$ or $h(y\theta)$ is reducible
by~$\longrightarrow_{R_2^{},ACh}^{}$.
Assume, without loss of generality, that
$h(x\theta)$ is reducible
by~$\longrightarrow_{R_2^{},ACh}^{}$.
By Lemma~\ref{h-terms}, $x\theta$ is either 0 or reducible by~$\longrightarrow_{R_2^{},ACh}^{}$.
If $x\theta$ is reducible by~$\longrightarrow_{R_2^{},ACh}^{}$, then $x\theta+y\theta$ is
also reducible by~$\longrightarrow_{R_2^{},ACh}^{}$, by~Corollary~\ref{plus-terms}.
If $x\theta = 0$, then clearly $x\theta+y\theta \; = \; 0 + y\theta$ is
reducible by~$\longrightarrow_{R_2^{},ACh}^{}$. \\

\ignore{Then $h(x\theta) \; \approx_{ACh}^{} \; l\beta$
for some~$l \rightarrow r$~in~$R_2^{}$. \\[-10pt]
Then either $x\theta = 0$
(i.e., the reduction uses the rule~$h(0) \rightarrow 0$) or
$x\theta$ is reducibleby~$\longrightarrow_{R_2^{},ACh}^{}$.
}

$q \lesssim p:$ Let $p = q \cdot \pi$. Let
$s_4^{} ~ = ~ s_1^{} |_q^{}$ and
$s_5^{} ~ = ~ s_4^{} \left[ \vphantom{b_a^c}
  h(x\theta) + h(y\theta) \right]_{\pi}^{}$.
Thus 
$s_4^{} \; \longrightarrow_{R_h^{},AC}^{\pi} \; s_5^{}$ and
$s_5^{}$ is reducible at the root by~$\longrightarrow_{R_2^{},ACh}^{}$.
By Corollary~\ref{R2ACh}, $s_5^{} \!\big\downarrow_{R_h^{}, AC}^{}$ is
reducible at the root by~$\longrightarrow_{R_2^{},AC}^{}$, i.e.,
$s_5^{} \!\big\downarrow_{R_h^{}, AC}^{} \; \approx_{AC}^{} \; l\beta$
for some left-hand side~$l$. Thus (putting it all together),
$s_1^{} |_q^{} \; \longrightarrow_{R_h^{},AC}^{!} \; s_5^{} \!\big\downarrow_{R_h^{}, AC}^{} \; \approx_{AC}^{} \; l\beta$.
Thus, again by Corollary~\ref{R2ACh}, $s_1^{} |_q^{}$ is reducible
by~$\longrightarrow_{R_2^{},ACh}^{}$.
\end{proof}


Finally, we can prove the convergence. 
\begin{appxlem}
$\longrightarrow_{R_2, ACh}^{} \,$ is $ACh$-convergent.
\end{appxlem}
\begin{proof}
  We prove by contradiction. Suppose $\longrightarrow_{R_2, ACh}^{} \,$ is not~$ACh$-convergent.
  Since $\longrightarrow_{R_2, ACh}^{} \,$ is terminating, there must
  be terms $s$ and $t$ such that $s$ and $t$ are in normal form modulo~$\longrightarrow_{R_2, ACh}^{} \,$
  and $s\;\not\approx_{ACh}\;t$.
  Since $\longrightarrow_{R_h, AC}^{} \,$ is \emph{AC}-convergent, the
  normal forms of $s$ and $t$ modulo~$\longrightarrow_{R_h, AC}^{} \,$
  cannot be $AC$-equal. Let $\widehat{s}$ and $\widehat{t}$ be respectively the
  normal forms of $s$ and $t$ modulo~$\longrightarrow_{R_h, AC}^{}$.
  If either $\widehat{s}$ or $\widehat{t}$ is reducible by~$\longrightarrow_{R_2, ACh}^{} \,$,
  then, by Lemma~\ref{commute}, either $s$ or~$t$ will be
  reducible by~$\longrightarrow_{R_2^{},ACh}^{}$.
  But since $\longrightarrow_{R_1^{}, AC}^{} \,$ is \emph{AC}-convergent,
  either $\widehat{s}$ or $\widehat{t}$ must be reducible by~$\longrightarrow_{R_1^{}, AC}^{}$. This
  also leads to a contradiction, since
  $\longrightarrow_{R_1^{},AC}^{} \; = \; \longrightarrow_{R_2^{},AC}^{} \; \cup \; \longrightarrow_{R_h^{},AC}^{}$.
\end{proof}


\label{appendix2}



\section{Asymmetric Combination Algorithm}
\begin{algorithm}
	\caption{Combination $\Delta_1 = ACUNh$ and $\Delta_2= F_{\Omega}$}
	\label{Alg:comb}
	\begin{algorithmic}
		\REQUIRE  $\Gamma_0$, the initial unification problem over the signature $\Sigma_1 \cup \Sigma_2$.
		\WHILE{there exist non-pure terms}
		\STATE {\bf Right abstraction:} For each alien subterm $t_1$ of $t$, let $x$ be
		a variable not occurring in the current system 
		and let $t'$ be the term obtained from replacing $t_1$ by $x$ in $t$. Then the
		original equation is replaced by two equations 
		$s =^{\downarrow} t'$ and $x =^{\downarrow} t_1$. 
		\STATE {\bf Left abstraction:} For each alien subterm $s_1$ of $s$ let $x$ be a
		variable not occurring in the current system 
		and let $s'$ be the term obtained from replacing $s_1$ by $x$ in $s$. Then the
		original equation is replaced by two equations
		$s' =^{\downarrow} t $ and $s_1 =^{\downarrow} x$.
		\ENDWHILE
		\WHILE{there exist non-pure equations}
		\STATE  {\bf Split non-pure equations}: Each non-pure equation of the form $s =^{\downarrow} t$
		is replaced
		by two equations $s =^{\downarrow} x$, $x =^{\downarrow} t$ where $x$ is always a new variable.
		The results is a system $\Gamma_2$ of pure equations.
		\ENDWHILE
		\STATE  {\bf Variable identification}: Consider all the possible partitions of the set of variables.
		Each
		partition produces a new system $\Gamma_3$ as follows. The variables in each class of the partition 
		are ``identified'' with each other by choosing an element of the class as a representative
		and replacing in the 
		system all occurrences of variables in each class by their representative.
		\STATE {\bf Choose an ordering and theory indices}: For each $\Gamma_3$ we consider all possible
		strict orderings~$<$ on the variables of the system and all mappings $ind$ from the set of
		variables into
		the set of indices $\{ 1,2 \}$. Each pair $(< , ind)$ yields a new system $\Gamma_4$.
		\STATE {\bf Split the system}: Each $\Gamma_4$ is split into two systems $\Gamma_{5,1}$ and
		$\Gamma_{5,2}$, the first containing only 1-equations and the second only 2-equations. 
		In the system $\Gamma_{5, i}$ the variables of index $j \neq i$ are treated as constants.
		Each $\Gamma_{5,i}$ is now a unification problem with linear constant restriction,
		where the linear ordering
		$<$ defines the set $V_c$ for each constant $c$ corresponding to an index $j \neq i$ variable. 
		
		\STATE {\bf Build $ACUNh$ automata}:
		For each $\Gamma_{5, 1}^{i}$ construct the corresponding automata 
		including the linear constant restriction and dis-equalities of the
		following form:
		\begin{itemize}
			\item $x \neq y$, if $x$ and $y$ are unique variables from $\Gamma_{4}$
			not equated by the partition. 
			\item $x \neq y$, if $x$ and $y$ are unique variables from $\Gamma_{4}$,
			$x \in \Gamma_{5, 1}$, and $y \in \Gamma_{5,2}$. 
			\end{itemize}  
			
			\STATE {\bf Solve $\Gamma_{5,1}^{i}$ and $\Gamma_{5,2}^{i}$}: For the initial system $\Gamma_0$ let 
			$\{ (\Gamma_{5,1}^{1} , \Gamma_{5,2}^{1}), \ldots , (\Gamma_{5,1}^{n} , \Gamma_{5,2}^{n}) \}$ be the
			output of the decomposition. 
			\FOR{$i = 1, \ldots , n$ and $j = 1,2$,}
			\IF{the automata for $\Gamma_{5, 1}^i$ have a non-empty intersection and the 
				syntactic unification algorithm on $\Gamma_{5,2}^i$ returns
				an injective and theory-preserving solution}
				\STATE return true.  
				\ENDIF
				\ENDFOR
				
				\end{algorithmic}
				
				\end{algorithm}

\end{document}